\documentclass[12pt,a4paper]{article}
\usepackage{amsmath,amstext,amssymb,amscd, color}
\usepackage{graphicx}

\usepackage{color,xcolor,mathdots}
\usepackage[centertableaux]{ytableau}
\usepackage[all]{xy}
\usepackage[margin=1in]{geometry}
\usepackage{tikz,tikz-cd}
\usetikzlibrary{shapes,arrows,svg.path}
\usepackage{booktabs}

\usepackage[colorlinks=true, pdfstartview=FitV, linkcolor=blue, citecolor=blue, urlcolor=blue]{hyperref}	

\newtheorem{thm}{Theorem}[section]
\newtheorem{prop}[thm]{Proposition}

\newtheorem{dfn}[thm]{Definition}
\newtheorem{rem}[thm]{Remark}
\newtheorem{ex}[thm]{Example}
\numberwithin{equation}{section}

\title{Hodge numbers for orbifolds of Calabi-Yau threefolds Fermat type and the Roan pairs}
\author{Sergei  Aleshin\thanks{
Institute for Information Transmission Problems, {19, Bol'shoi Karetnii}, {127051}, {Moscow}, {Russia, aless2001@mail.ru
}}, Alexander  Belavin\thanks{Landau Institute for Theoretical Physics, 142432 Chernogolovka, Russia, sashabelavin@gmail.com} and Gleb Koshevoy\thanks{Institute for Information Transmission Problems, {19, Bol'shoi Karetnii}, {127051}, {Moscow}, {Russia, koshevoyga@gmail.com
}} 
}

\begin{document}
\maketitle

\begin{abstract}
First, for orbifolds of Calabi-Yau threefolds of Fermat type,  we define Roan's Hodge numbers. We prove, Theorem \ref{main} , that for all orbifols of Calabi-Yau threefolds Fermat type,  Roan's Hodge numbers correctly count the stringy Euler numbers due to Vafa formula.
Second,
  for Calabi-Yau threefolds Fermat type, we apply Roan's Hodge numbers  to get a relation between the Borcea-Voisin construction \cite{Borcea, Voisin}  and Berglund-H\"ubsch-Krawits mirror symmetry (BHK mirror symmetry).
   
 Third,
   for orbifolds of K3 surfaces Fermat type,  we establish relations twisted and untwisted parts of the second mixed cohomology to deformations and the Roan pairs.  This allow us to confirm the BHK mirror symmetry for such orbifolds, since the Euler numbers computed by the Vafa formula are
   equal $24$ for all of them.  

\end{abstract}

\section{Introduction}

 Calabi-Yau (CY) manifolds/orbifolds have always occupied a special place in geometry
and physics. The combinatorial nature of these objects make them a fertile ground to test
new ideas and techniques.

{
For physics Calabi-Yau threefolds are of main interest in relations to string theory}. A classical result \cite{C-W} states that compactification of $E_8\times E_8$ heterotic string on a Calabi-Yau threefold $X$ yield a four-dimensional $\mathcal N=1$ effective theory with%

\begin{enumerate}

    \item gauge group $E_6\times E_8$ (for the standard embedding);
    \item complex structure moduli $h^{2,1}(X)$  of chiral supermultiplets in the $\mathbf{27}$ of $E_6$ and K\"ahler moduli $h^{1,1}(X)$ vector-like supermultiplets in $\overline{\mathbf{27}}$;
    \item $|\chi(X)/2|$ is number of generations, chiral matter supermultiplets.
    
\end{enumerate}

In \cite{Belavin} it was proposed the  free fields construction to heterotic string on orbifolds of  Calabi-Yau threefolds of Berglund-H\"ubsch type using 
combinatorial  Borisov's approach to the vertex algebra of  $\mathcal N=2$ super Virasoro.

{

The Vafa formula provides the mathematically correct invariant orbifolds, stringy Euler numbers, compatible with string theory. Stringy Hodge numbers split into untwisted sector (Standard Cohomology) and twisted sectors (Singular Loci States), and Mirror symmetry swap stringy Hodge numbers.

There are twofolds of our paper. 

First, for orbifolds of Calabi-Yau threefolds of Fermat type,  we define Roan's Hodge numbers. Namely,  for the untwisted sector, we leave standard cohomology and they  correspond to  invariant  deformations. 
The twisted sectors  is based on the notion of Roan's pairs. The cardinality of the union of such sectors, we call {\em Roan's Hodge numbers}.
We obtain an explicit algorithmic construction such Hodge numbers.

Theorem \ref{main}  states that, for all orbifols of Calabi-Yau threefolds Fermat type,  Roan's Hodge numbers correctly count the stringy Euler numbers due to Vafa formula.}

Second,
  for Calabi-Yau threefolds Fermat type, we apply Roan's Hodge numbers  to get a relation between the Borcea-Voisin construction \cite{Borcea, Voisin}  and Berglund-H\"ubsch-Krawits mirror duality.
 {
 
Recall, that Borcea and Voisin \cite{Borcea, Voisin}  consider a threefold being the product of the elliptic curve $E$, $x_1^4+x_2^4=x_0^2$ with involution 
$j:E\to E$ in the weighted projective space $P(1,1,2)$, 
and a K3 surface $S$, $x_0^2=f(w, y,z)$ with involution $\iota$ acting by $-1$ on the lattice of second cohomology group $H^{2,0}(S)$.
They proved that after resolution of singularities of $E\times S/j\times \iota$, the threefold
$X=\widetilde{E\times S/j\times \iota}$ is a Calabi-Yau threefold.
Moreover, the Borcea-Voisin construction provides such Calabi-Yau threefold $X$ with a mirror $X^{\vee}$ of the same elliptic curve form but with $K3$ with swapped Nikulin's parameters.

We are interested in BV-construction  for $K3$ surface of Fermat type. There are ten such $K3$ surfaces. For each of such ten types K3 surfaces $x_0^2=f(w, y,z)$ of Fermat type, we consider the Fermat threefold  
\[
W=x_1^4+x_2^4+f(w, y,z)
\]
and  the orbifold ${W/G}$, where the group 
$G$ is isomorphic to $G_0\times Z_2$, where $G_0$ is the cyclic group of order $d$, where $d$ is the degree of the Fermat threefold, and $Z_2$ is the Abelian group of order 2.

Denote by $\tilde X=\widetilde{W/G}$ \, a  Calabi-Yau threefold  obtained by resolution of singularities ${W/G}$.

We prove in Theorem \ref{BV-BHK} existence of a group $G\cong G_0\times Z_2$ such that the Roan's Hodge numbers of  $\widetilde{W/G}$ coincide with Hodge numbers of $X=\widetilde{E\times S/j\times \iota}$. The proof is going case-by-case, since it might be several groups isomorphic 
 $G_0\times Z_2$ with different Hodge numbers and  Euler numbers.}

There exists a database \cite{KS} for reflexive polytopes.

We provide a database \cite{DB} for all  orbifolds of Calabi-Yau threefolds Fermat type indicating decomposition Hodge numbers into deformations and Roan' pairs.

For a tuple of positive integers $(a_1\le  a_2\le  a_3\le  a_4\le a_5)$ such that 
\[
\frac{1}{a_1}+\cdots  +\frac{1}{a_5}=1,
\]
the Fermat threefold 
\[
W_0=x_1^{a_1}+x_2^{a_2}+x_3^{a_3}+x_4^{a_4}+x_5^{a_5}, 
\]
of degree $d$, where $d=ta_5$, where $t$ is minimal such that $\frac{ta_5}{a_i}\in \mathbb Z$,  in the weighted projective space
\[
\mathbb P(k_1, k_2, k_3, k_4, k_5), \quad k_i=\frac{d}{a_i}, \quad i=1, \cdots, 5, 
\]
quotient by the  group $G_0$ a cyclic group of order $d$ is a quasi-smooth Calabi-Yau threefolds.

Such Calabi-Yau threefolds and their orbifolds are of our interest here.

Specifically, for $G^{\max} \subset SL_5$ a maximal group of diagonal symmetries of $W_0$ and an admissible  subgroup $G$,  $G_0\subset G\subset G^{\max}\subset SL_5$, we consider an orbifold $W_0/G $. Such an orbifold  is a singular variety in general, and we are interested in desingularization of such orbifold and an efficient algorithm to obtain its Hodge numbers $h^{1,1}$ and $h^{2,1}$, and moreover to get sets of lattice elements which corresponds to deformations of Complex structures and whose come from resolution singularities.

Krawitz \cite{Krawitz} invented orbifold Hodge numbers and formulated 
the Berglund-H\"ubsch mirror symmetry as a mirror pair construction with respect to such Hodge numbers.  Namely  the  {\em  Berglund-H\"usch-Krawitz duality} associates a Landau-Ginzburg orbifold $W/G$, where $W$ is an invertible nondegenerate quasihomogeneous polynomial potential and $G$ is an admissible group, with a dual Landau-Ginzburg orbifold $W^T /G^T$ .

For Calabi-Yau threefolds of Fermat type we provide a combinatorial view point on the mirror duality following the ideas of Roan \cite{Roan}.
Namely, let $\widetilde{W/G}$ be a smooth variety obtained by resolution singularities of ${W/G}$. 

Vafa \cite{Vafa} proposed a formula for the Eulerian characteristics of $\widetilde{W/G}$
\begin{equation}\label{vafa0}
   \chi(\widetilde{W/G})=\frac{1}{|G|}\sum_{g_1,g_2\in G}(\prod_{i\in [5], \mbox{ such that i-th coordinate }g_1 \mbox{ and } g_2=1}(1-a_i)).
\end{equation}

The mirror symmetry implies the following equality
\begin{equation}\label{mir0}
 \chi(\widetilde{W/G})=-\chi(\widetilde{W/G^T}).   
\end{equation}

From the Vafa formula (\ref{vafa0}) we can not directly see the equality (\ref{mir0}).

Roan (\cite{Roan}) showed that (\ref{mir0}) can be computed using deformations of the complex structure wrt $G$, deformations of the complex structure wrt $G^T$ and the Roan pairs, some pairs of orthogonal elements in $G^T$ and $G$ and some pairs of orthogonal elements of $G$ and $G^T$. Since for  $\chi(\widetilde{W/G^T})$, one has to switch all these data from $G$ and $G^T$,  we
obtain (\ref{mir0}) (see details in Section \ref{vafa-1}).

For $a_i\ge 3$, $i=1, \ldots, 5$,
Roan proved the coincidence between the Vafa formula and a combinatorial computation with Roan's pairs (see \cite{Roan} Theorem 1 ).

We checked the Roan formula for all $147$ Calabi-Yau Fermat and their orbifolds, including 49 cases Fermat and their orbifolds have been considered by 
Roan.  





{
We do computer verification that twice the difference of the Hodge numbers obtained by the Roan recipe to count twisted parts of Hodge numbers gives correct Euler characteristic computed by the Vafa formula.}

An advantage of this formula it that we get explicitly set of lattice vectors cardinality of which is counted by the Hodge numbers 
and these sets are of importance for applications for the  Heterotic string with Calabi-Yau compactification \cite{Belavin}.
Namely, we can  get Calabi-Yau multiplier of any total vertex operator of Heterotic string without checking the OPE axioms. 

For all orbifolds of quintic (up to symmetry isomorphisms), we give lists of the deformations and  
Roan's pairs for counting the Hodge numbers.

In Section 7 we consider the Borcea-Voisin construction (Voisin \cite{Voisin} and Borcea \cite{Borcea})  and relate it to orbifolds of corresponding Fermat threefolds and get another mirrors due to BHK. There are ten of fourteen $K3$ surfaces of Fermat type, which fit into the BV-construction.
For each of such cases, we find a corresponding Fermat orbifold, such that it desingularization has the same Hodge numbers. For cases of self mirrors in BV-construction, we found that BHK orbifolds are non self mirror. 
{
This reminds the case with elliptic curve. An elliptic curve $E$ is self dual by replacing $\tau$ to $-1/\tau$, but, from the orbifold view point, we have two different orbifolds.}

In Section 9, we define the Roan pairs for K3 surfaces Fermat type (there are 14 such K3 surfaces) and demonstrate in Theorem \ref{K3Roan}  decomposition of second mixed cohomology into untwisted and twisted parts  using deformations and Roan's pairs.
This allows us to confirm BHK mirror symmetry, since $h^{1,1}=20$ for all such orbifolds and the Hodge diamond is the same for all orbifolds.

In Conclusion we note on a possible relations to F-theory. 

The algorithm in \cite{AB} has been used for obtaining all orbifolds of Fermat threefolds. 

\subsection{Acknowledgments} The research of S. Aleshin was carried out within the state assignment of Ministry of Science and Higher Education of the Russian Federation for IITP RAS,  A.Belavin has been supported by the state assignment FFWR-2024-0012, the work of G. Koshevoy was carried out at the Center for Pure Mathematics at MIPT with financial support by RSF grant 26-11-00115.

\section{Calabi-Yau Fermat orbifolds }
\subsection{Fermat hypersurface, Fano and Calabi-Yau}
For a solution $(a_1\le  \cdots \le  a_m)\in {\mathbb N}^m$ of a classic problem in Egyptian Fractions
\begin{equation}\label{eqypt1}
    \frac{1}{a_1}+\cdots +\frac{1}{a_m}=1,
\end{equation}
we define a degree $d=ta_m$, where $t\in\mathbb N$ is  a minimal positive integer such that, for any $j=1,\cdots, m$,
there holds
\[
\frac{ta_m}{a_j}\in \mathbb N.
\]
Then the equation $W_0=0$, 
\begin{equation}
   W_0= \sum x_i^{a_i}
\end{equation}
defines the Fermat hypersurface, a hypersurface of degree $d$   in the weighted projective space 
$\mathbb P(k_1,\cdots, k_m)$, $k_1\ge \cdots \ge k_m=t$,
\[
k_i=\frac{ta_m}{a_j}.
\]

Because of (\ref{eqypt1}), we have 
\[
\sum k_i=ta_m=d.
\]
Hence the weighted projective space 
$\mathbb P(k_1,\cdots, k_m)$, $k_1\ge \cdots \ge k_m=t$ is a {\em Fano } toric variety 
because the fractions $\frac{d}{k_j}=a_j$, $j=1,\cdots, a_{m-1}$, $\frac{d}{k_m}=t$,
are positive integers.

Moreover, quotient of such a Fermat hypersurface by the symmetry $G_0$ group of $\mathbb P(k_1,\cdots, k_m)$, 
has a unique singularity at zero and is a Calabi-Yau variety. 
\subsection{Automorphisms}
We consider $G(W_0)$, the full group of automorphism of $W_0$, what is the product of cyclic groups
\[
G(W_0):=Z_{a_1}\times \cdots \times Z_{a_m},
\]
and the maximal group, a subgroups which preserves  the monomials $\prod_{i=1}^m x_i$, or $G^{\max}$ is a diagonal subgroup of $SL_m$ of symmetries of $W_0$, that is the intersection 
\[
G^{\max}=G(W_0)\cap SL_m.
\]
For a lattice vector $z=(z_1, \cdots, z_m)\in\mathbb Z^m$, we define 
\[
\mathbf{e}_{a_1, \ldots, a_m}(z):=(exp(\frac{2\pi\sqrt{-1}z_1}{a_1}), \cdots, exp(\frac{2\pi\sqrt{-1}z_m}{a_m})).
\]
When we fixed the exponents $a_1, \ldots , a_m$ of the Fermat type, we freely use $\mathbf e$. Therefore, we may regard the group $G^{\max}$ as a set of lattice points $(z_1, \ldots, z_m)$ of the following set
\begin{equation}\label{group1}
Z(a_1, \cdots, a_m):=\{z\in\mathbb Z^m\,|\, 0\le z_i\le a_i-1\mbox{ and }
\sum_i\frac{z_i}{a_i}\in \{0,1,\cdots, m-1\}\}.
\end{equation}

Namely, we have 
\[G^{\max}:=\{\mathbf e(z), z\in Z(a_1, \cdots, a_m)\}.\]

An element  $\mathbf{e}_{a_1, \ldots, a_m}(z)$ is {\em of level} $k$ if it holds true
$\sum_i\frac{z_i}{a_i}=k$.

The cyclic subgroup $G_0\subset G^{\max}$, $G_0=Z_d$, is determined by the generator
\[
g_0:=\mathbf{e}(1,1,1,1,1)=(exp(\frac{2\pi\sqrt{-1}}{a_1}), \cdots, exp(\frac{2\pi\sqrt{-1}}{a_m})).
\]
\subsection{Subgroups and mirror dual subgroups}

We are interested in subgroups $G_0\subset G\subset G^{\max}$ and their {\em mirror dual by Krawitz} 
$G_0\subset G^T\subset G^{\max}$. The latter means that, for any pair of elements $g=\mathbf e(u_1, \ldots, u_m)\in G$ and $g'=\mathbf e(v_1, \ldots, v_m)\in G^T$, there holds 
\begin{equation}\label{duality-1}
\sum_i \frac 1{a_i}u_iv_i
\in \mathbb Z,\mbox{ or equivalently  } \sum_i g_ia_ig_i' \in \mathbb Z.
\end{equation}
 

For example, $G_0$ and $G^{\max}$ form a mirror dual pair.

\subsection{Deformations}
The set of the exponents of monomials  of all deformations is 
the following subset of $Z(a_1, \cdots, a_m)$,  
\[
S_{a_1, \cdots, a_m}:=\{S\in\mathbb Z^m\,|\, 0\le S_i\le a_i-2\mbox{ and }
\sum_i\frac{S_i}{a_i}=1\}.
\]
Then the set  of monomials
\[
\{x^S, \,\, S\in S(a_1, \cdots, a_m)\}
\]
is the set of all {\em deformations of the complex structure on} $W_0$.

Each monomial of such a set is invariant under the action of $G_0$,
\[
g_0(x^S)=x^S \cdot exp(2\pi\sqrt{-1}(\sum_i \frac{S_i}{a_i}))=x^S,
\]
and, for a power of the generator $g_0$, we have 
\[
g_0^b(x^S)=x^S \cdot exp(2\pi\sqrt{-1}(b-\sum_{j: b/a_j\in \mathbb N}S_j))=x^S.
\]
In fact the degree of the exponent is  equal $b\sum_i \frac{S_i}{a_i}-\sum_{j: b/a_j\in \mathbb N}S_j$, and this number is integer since $b$, $sum_i \frac{S_i}{a_i}$ and $\sum_{j: b/a_j\in \mathbb N}S_j$ are integer.

\subsection{}
For a subgroup $G_0\subset G\subset G^{\max}$, we define the invariant  subset of monomials under the action $G$, 
$S^G\subset S_{a_1, \cdots, a_m}$. 

It is convenient to regard the image of the deformations in the mirror group 
\begin{equation}
    {\mathbf e}(S_{a_1, \cdots, a_m}):=\{{\mathbf e}(S_1, \cdots, S_m)\subset G^T\}.
\end{equation}
Then we have the following 
\begin{dfn}
The set \[S^G= {\mathbf e}^{-1}(G^T\cap {\mathbf e}(S_{a_1, \cdots, a_m}))\]
is the set of $G$-invariant deformations.
\end{dfn}
One can see that there holds  $S^G= {\mathbf e}^{-1}(G^T)\cap S_{a_1, \cdots, a_m}$. For example,  
for the group $G^{\max}$, we have 
$$
S^{G^{\max}}= \mathbf e^{-1}(G_0)\cap S_{a_1, \ldots, a_m}. 
$$

\begin{rem}
    One can regard the set of all deformations as a subset of lattice points of the simplex 
    \[
\hat B(k_1,\cdots, k_m):=\{z_i\ge 0, \sum_i {k_iz_i}=d\},
\]
since it is nothing but  the simplex
\[
\Delta(a_1,\cdots, a_m):=\{(z_1, \cdots , z_m), \, z_i\ge 0, \sum_i \frac{z_i}{a_i}=1\},
\]
  and the latter simplex is obtained by shifting on vector $(\frac{1}{a_1}, \cdots, \frac{1}{a_m})
$, the reflexive polytope corresponding $W_{a_1, \cdots , a_m}$ in the weighted projective space $\mathbb P(k_1, \cdots, k_m$,
\[
B(k_1,\cdots, k_m):=\{z_i\ge -1, \sum_{i=1}^m {k_iz_i}=0\}.
\]

Because of that, and having in mind the mapping $\mathbf e$, one can regard elements of level one $G$ and $G^T$ as some sets of lattice points of simplexes $B(k_1,\cdots, k_m)$ or $\hat B(k_1,\cdots, k_m)$.


\end{rem}

\section{Vafa formula for Euler numbers of Fermat orbifolds}\label{vafa-1}
\subsection{}
Consider Calabi-Yau Fermat hypersurface 
\[
W:=W_{a_1,\cdots ,a_m}=\sum x_i^{a_i}
\]
in the weighted projective space $P(k_1,\cdots, k_m)$.


For a group $G^{\max}\supset G\supset G_0$, we consider $G$ as a diagonal subgroup of $SL_m$.

For $\widetilde{ W/G}$, a manifold obtained by resolutions of singularities an orbifold $W/G$, 
the Vafa formula \cite{Vafa} reads as follows

\begin{equation}\label{vafa2}
   \chi(\widetilde {W/G})=\frac{1}{|G|}\sum_{g_1,g_2\in G}(\prod_{i\in [m], \mbox{ such that i-th coordinate }g_1 \mbox{ and } g_2=1}(1-a_i)).
\end{equation}

We set the product over the empty set to be equal $1$.

Here we give simplest examples with an elliptic curve.

Consider $m=3$. In such a case, there are three solutions to (\ref{eqypt1}),

\subsection{}\label{E1}
For $a_1=a_2=a_3=3.$, 
we get the cubic curve
\[
W_0=x_1^3+x_2^3+x_3^3
\]
in $\mathbb P(1,1,1)$.

The group $G$ is the product of three cyclic groups $Z_3$ of order $3$,
$G^{\max}$ is product  of two $Z_3$ and has the following elements,
\[
g_1=\begin{pmatrix} 1&0&0\cr 0&1&0\cr 0&0&1
\end{pmatrix}, 
g_2=\begin{pmatrix} \omega &0&0\cr 0&\omega &0\cr 0&0&\omega
\end{pmatrix}, 
g_3=\begin{pmatrix} \omega^2 &0&0\cr 0&\omega^2 &0\cr 0&0&\omega^2
\end{pmatrix}, 
g_4=\begin{pmatrix} 1 &0&0\cr 0&\omega &0\cr 0&0&\omega^2
\end{pmatrix}, 
g_5=\begin{pmatrix} 1 &0&0\cr 0&\omega^2 &0\cr 0&0&\omega
\end{pmatrix}, 
\]
\[
g_6=\begin{pmatrix} \omega &0&0\cr 0&1 &0\cr 0&0&\omega^2
\end{pmatrix}, 
g_7=\begin{pmatrix} \omega^2 &0&0\cr 0&1 &0\cr 0&0&\omega
\end{pmatrix}, 
g_8=\begin{pmatrix} \omega &0&0\cr 0&\omega^2 &0\cr 0&0&1
\end{pmatrix}, 
g_9=\begin{pmatrix} \omega^2 &0&0\cr 0&\omega &0\cr 0&0& 1,
\end{pmatrix}.
\]
$\omega$ is a root of unity of order $3$.
$G_0$ is $Z_3$ and has elements
\[
g_1=\begin{pmatrix} 1&0&0\cr 0&1&0\cr 0&0&1
\end{pmatrix}, 
g_2=\begin{pmatrix} \omega &0&0\cr 0&\omega &0\cr 0&0&\omega
\end{pmatrix}, 
g_3=\begin{pmatrix} \omega^2 &0&0\cr 0&\omega^2 &0\cr 0&0&\omega^2
\end{pmatrix}.
\]
Let us check that the Vafa formula (\ref{vafa2}) gives zero for $W_0/G^{\max}$ and $W/G_0$ and they form the mirror pair.
Here are summands of (\ref{vafa2}) for each pair $(g_i,g_j)$,
\[
\begin{pmatrix}
    \phantom{g} & g_1& g_2& g_3& g_4& g_5&   g_6& g_7& g_8& g_9\cr  
    g_1&  -8& 1& 1& -2& -2& -2& -2& -2& -2\cr
     g_2&  1& 1& 1& 1& 1& 1& 1& 1& 1\cr
      g_3&  1& 1& 1& 1& 1& 1& 1& 1& 1\cr
      g_4&  -2& 1& 1& -2& -2& 1& 1& 1& 1\cr
      g_5&  -2& 1& 1& -2& -2& 1& 1& 1& 1\cr
      g_6&  -2& 1& 1& 1& 1& -2& -2& 1& 1\cr
      g_7&  -2& 1& 1& 1& 1& -2& -2& 1& 1\cr
      g_8&  -2& 1& 1& 1& 1& 1& 1& -2& -2\cr
      g_9&  -2& 1& 1& 1& 1& 1& 1& -2& -2\cr
\end{pmatrix}
\]
The sum of entries is equal $0$.

This is the unique mirror pair for the cubic since  $G^{\max}/G_0=Z_3$, and the later group has only two subgroups, $e$ and $Z_3$.
\subsection{}\label{E2}
The second Fermat  curve  is
\[
x_1^2+x_2^4+x_3^4\mbox{ in } \mathbb P(2,1,1).
\]
Let us check by the Vafa formula, that the mirror pair in such a case also with Euler characteristics $0$.

The maximal group has the following elements, we write only diagonals
\[
h_0=(0,0,0),\,h_1=(1/2,1/4,1/4),\,h_2=(0,1/2,1/2),\,h_3=(1/2,3/4,3/4),\, 
\]
\[h_4=(1/2,1/2,0),\,h_5=(1/2,0, 1/2),\, h_6=(0,1/4,3/4),\, h_7=(0,3/4,1/4).
\]
In such notations, the Vafa formula has to be read with zeros instead of units.
For $G_0$, we get $\chi(W_0/G_0)=0$, since the summands are 
\[
\begin{pmatrix}
    \phantom{g} & h_0& h_1& h_2& h_3& \cr  
    h_0&  -9& 1& -1& 1\cr
     h_1&  1& 1& 1& 1&\cr
      h_2&  -1& 1& -1& 1\cr
      h_3&  1& 1& 1& 1& \cr
      \end{pmatrix}
\]
and, for  $G^{\max}$, we also get $\chi(W_0/G^{\max})=0$,
\[
\begin{pmatrix}
    \phantom{g} & h_0& h_1& h_2& h_3& h_4& h_5& h_6 &h_7\cr  
    h_0&  -9& 1& -1& 1&-3 & -3& -1&- 1\cr
     h_1&  1& 1& 1& 1&\ 1& 1& 1& 1&\cr
      h_2&  -1& 1& -1& 1 &1& 1& -1& -1&\cr
      h_3&  1& 1& 1& 1& 1& 1& 1& 1&\cr
      h_4&  -3& 1& 1& 1& -3& 1& 1& 1&\cr
      h_5&  -3& 1& 1& 1& 1& -3& 1& 1&\cr
      h_6&  -1& 1& -1& 1& 1& 1& -1& -1&\cr
      h_7&  -1& 1& -1& 1& 1& 1& -1& -1&\cr
      
      \end{pmatrix}
\]

\subsection{}\label{E3}
The last case of the Calabi-Yau  Fermat curve is
\[
x_1^2+x_2^3+x_3^6 \mbox{ in }\mathbb P(3,2,1).
\]
It is also elliptic curve and has the Euler number $0$.
For $G_0$, we have the following elements, we again write only diagonals
\[
h_0=(0,0,0),\,h_1=(\frac{1}{2},\frac{1}{3},\frac{1}{6}),\,h_2=(0,\frac{2}{3},\frac{1}{3}),\,h_3=(\frac{1}{2},0,\frac{1}{2}),\, h_4=(0,\frac{1}{3},\frac{2}{3}),\,h_5=(\frac{1}{2},\frac{2}{3}, \frac{5}{6}),
\]
The entries in the Vafa formula are
\[
\begin{pmatrix}
    \phantom{g} & h_0& h_1& h_2& h_3& h_4& h_5\cr  
    h_0&  -10& 1& -1& -2&-1 & 1\cr
     h_1&  1& 1& 1& 1&\ 1& 1\cr
      h_2&  -1& 1& -1& 1 &-1& 1\cr
      h_3&  -2& 1& 1& -2& 1& 1\cr
      h_4&  -1& 1& -1& 1& -1& 1\cr
      h_5&  1& 1& 1& 1& 1& 1\cr
      \end{pmatrix}
      \]
This matrix has the elements sum equal zero.

\section{ Calabi-Yau Fermat threefolds}

Consider a Calabi-Yau Fermat threefold 
\[
W_{a_1,a_2,a_3,a_4,a_5}=\sum x_i^{a_i}
\]
in the weighted projective space $P(k_1,\cdots, k_5)$.

There are $147$ in total solutions to 
\[
\sum_{i=1}^5\frac{1}{a_i}=1, \qquad a_1\le a_2\le a_3\le a_4\le a_5.
\]

The degree $d$ of $W_{a_1,\cdots, a_5}$ satisfies 
\[d=\sum_ik_i=ta_5,\]
where $t\in \mathbb N$ is minimal such that $\frac{ta_5}{a_i}\in \mathbb N$, $i=1, \ldots, 4$.
Obvioulsy, there holds 
\[
d=k_ia_i,\quad \forall\, i=1,\cdots, 5.
\]
The  group $G_0$ is a cyclic group $Z_d$ with a generator $g_0$ being  matrix with the diagonal 
\[
g_0:=\mathbf e(1,1,1,1,1)=(exp(\frac{2\pi\sqrt{-1}}{a_1}, \cdots, exp(\frac{2\pi\sqrt{-1}}{a_5})).
\]

Then the Vafa formula for the Euler number for $W_{a_1,\cdots,a_5}$
\begin{equation}\label{vafa3-1}
   \chi(W/G_0)=\frac{1}{|G_0|}\sum_{0\le p,q\le d-1}(\prod_{i\in [5], \mbox{ such that i-th coordinate }g^p \mbox{ and } g^q=1}(1-a_i)).
\end{equation}


For an orbifold $\widetilde {W/G}$, obtained by resolution of singularities $W/G$, we have 

\begin{equation}\label{vafa3-2}
   \chi(\widetilde{W_{a_1,\ldots, a_5}/G})=\frac{1}{|G|}\sum_{g_1,g_2\in G}(\prod_{i\in [5], \mbox{ such that i-th coordinate }g_1 \mbox{ and } g_2=1}(1-a_i)).
\end{equation}

For the pair of mirror dual groups by Krawitz, the mirror symmetry holds in the form
\[
 \chi(\widetilde{ W_{a_1,\ldots, a_5}/G})=- \chi(\widetilde{ W_{a_1,\ldots, a_5}/G^T}).
\]
\section{Roan's pairs}
\subsection{Roan's pairs and main theorem}
The cardinality of the set of deformations $S^G$ equals $h^{2,1}( ( W_{a_1,\ldots, a_5}/G))$, and the cardinality of the set of deformations $S^{G^T}$ equals $h^{2,1}( W_{a_1,\ldots, a_5}/G^T))$.

Since $ W_{a_1,\ldots, a_5}/G$ is a singular orbifold, the difference 
\[h^{2,1}( W/G^T)- h^{2,1}( W/G)\] might be not equal to $\frac 12\chi_{Vafa}( W/G)$. That is the case, for almost all orbifolds of Calabi-Yau Fermat threefolds.  

Resolution of singularities brings additional terms to $h^{2,1}( W/G^T))- h^{2,1}( W/G))$.   

We call {\em Roan's pairs} the additional terms to the twisted sector of Hodge numbers. Roan's pairs corresponding to  the twisted sector.

To define $(G^T, G)-$ Roan's pairs,  
let us consider $G^T_{2Z}\subset G^T$  the subset of level one elements of $G^T$ which have two zeros among its coordinates.
And  $G_{3Z}\subset G$ is the subset of level one elements of  $G$  with three zeros among coordinates. 

\begin{dfn}
Then a pair  $(\mathbf z', \mathbf z)\in G^T_{2Z}\times G_{3Z}$ is a $(G^T, G)$- {\em Roan's pair} if 

\[
\sum  z'_i\,a_i\, z_i=0. 
\]

A pair   $(\mathbf z, \mathbf z')\in G_{2Z}\times G^T_{3Z}$ is a $(G, G^T)-$ {\em Roan's pair} if 
\[
\sum z_i\,a_i \,z'_i=0.
\]

     For any Fermat CY threefold and a group $G_0\subset G\subset G^{\max}$, we set 
\begin{equation}\label{R1}
h^{2,1}(\widetilde{W/G})=h^{2,1}( W/G) +|\{ (G^T, G) \mbox{ Roan's pairs}\}|=h^{1,1}(\widetilde { W/G^T});
\end{equation}

\begin{equation}\label{R2}
h^{2,1}(\widetilde{W/G^T})=h^{2,1}( W/G^T) +|\{ (G, G^T) \mbox{ Roan's pairs}\}|=h^{1,1}(\widetilde { W/G});
\end{equation}
\end{dfn}

Due to the above definition, we have (\ref{R1}) and (\ref{R2}) swap under  swap $G$ and $G^T$ as it has to be for the mirror symmetry. Moreover, the Euler characteristic with respect to such Hodge numbers coincides with Euler characteristic due to the Vafa formula.

\begin{thm}\label{main}
For Roan's Hodge numbers  (\ref{R1})-(\ref{R2}), there holds
\[
\frac 12 \chi(\widetilde{W/G})= 
h^{2,1}(\widetilde{ W/G^T}))- h^{2,1}( \widetilde{W/G})),
\]
where $ \chi(\widetilde{W/G})$ is the Euler characteristic computed by the Vafa formula (\ref{vafa3-1}).
\end{thm}
{\em Proof}. We verified this formula by numeric computations. We put details of computations in \cite{DB}.
\hfill $\Box$

\begin{rem}
    One of the main reason to get a proof by explicit computations deformations and Roan's pairs is to obtain a database for Hodge numbers for all orbifolds of CY-threefolds of Fermat type. 
\end{rem}
\begin{rem}
It is worth noting that these combinatorial structures admit a natural physical interpretation within the free-field construction of the heterotic string.  As shown in \cite{MBG}, for Fermat-type orbifolds the Hodge numbers can be extracted directly from the structure of massless vertex operators corresponding to the $\mathbf{27}$ and $\overline{\mathbf{27}}$ representations of $E_6$. In this framework, certain massless states involve pairs of lattice vectors which (up to the generator of $G_0$) coincide precisely with the $(G^T,G)$ and $(G,G^T)$ Roan pairs, thus providing a physical realization of Roan's combinatorial definition.
\end{rem}
\subsection{Roan's pairs and lattice vectors}
For a  $(G^T, G)$- Roan's pair 
\[
(\mathbf z', \mathbf z), \quad \mathbf z'\in G^T_{2Z},\, \mathbf z\in G_{3Z}, 
\]
let consider the following pair of  vectors 
\[
(\mathbf Z', \mathbf Z)=(\mathbf z'+(\frac{a_1-1}{a_1},\cdots, \frac{a_5-1}{a_5}) \,\mod \mathbb Z, \mathbf z+(\frac{a_1-1}{a_1},\cdots, \frac{a_5-1}{a_5}) \,\mod \mathbb Z).
\]
Then, we have the following

{\bf Lemma}. {\em For a  $(G^T, G)$- Roan's pair 
\[
(\mathbf z', \mathbf z), \quad \mathbf z'\in G^T_{2Z},\, \mathbf z\in G_{3Z}, 
\]
we have 
\[\sum Z'_i=2, \quad \sum Z_i=3.\]}
{\em Proof}. Is obvious. 
\hfill $\Box$





    
\section{Orbifolds of the quintic}
Consider the hypersurface of degree 5
\[
W_0=x_1^5+\cdots +x_5^5
\]
in $\mathbb P(1,\cdots,1)$. Following \cite{GP}, we consider four different orbifolds and their mirrors, all other are 
obtained by action of the permutations groups $S_5$ on $(x_1, \ldots, x_5)$. For each of such an orbifold we compute  deformations and Roan's pairs. 
\begin{enumerate}
    \item 
For $G_0=Z_5$, $h^{2,1}(W_0/G_0)=101$, $h^{2,1}(W_0/G^{\max})=1$, there are $101$ $G$-incariant deformations, one $G^T$-deformation and there are no Roan's pairs, and 
\[
\chi(W_0/Z_5)=-200.
\]

\item For $G_1=Z_5\times Z_5$ with generators
[(0, 1/5, 0, 0, 4/5), (1/5, 1/5, 1/5, 1/5, 1/5)]
and dual group $G^T_1=Z_5\times Z_5\times Z_5$ with generators 
[(0, 0, 1/5, 4/5, 0), (0, 1/5, 0, 3/5, 1/5), (1/5, 1/5, 1/5, 1/5, 1/5)],
we have 
 $25$ $G_1$-invariant deformations

(1, 1, 1, 1, 1) 
(0, 0, 2, 3, 0) (0, 0, 3, 2, 0) (0, 1, 0, 3, 1) (0, 1, 1, 2, 1) 
(0, 1, 2, 1, 1) (0, 1, 3, 0, 1) (0, 2, 0, 1, 2) (0, 2, 1, 0, 2) 
(1, 0, 1, 3, 0) (1, 0, 2, 2, 0) (1, 0, 3, 1, 0) (1, 1, 0, 2, 1) 
(1, 1, 2, 0, 1) (1, 2, 0, 0, 2) (2, 0, 0, 3, 0) (2, 0, 1, 2, 0) 
(2, 0, 2, 1, 0) (2, 0, 3, 0, 0) (2, 1, 0, 1, 1) (2, 1, 1, 0, 1) 
(3, 0, 0, 2, 0) (3, 0, 1, 1, 0) (3, 0, 2, 0, 0) (3, 1, 0, 0, 1) 

and five $G^T_1$-invariant deformations

(1, 1, 1, 1, 1) 
(0, 2, 0, 0, 3) (0, 3, 0, 0, 2) (1, 0, 1, 1, 2) (1, 2, 1, 1, 0);

There are no $(G_1,G^T_1)$-Roan's pairs and there are $24$ $(G^T_1,G_1)$-pairs

[(1/5, 0, 1/5, 3/5, 0), (0, 1/5, 0, 0, 4/5)] 
[(1/5, 0, 1/5, 3/5, 0), (0, 2/5, 0, 0, 3/5)] [(1/5, 0, 1/5, 3/5, 0), (0, 3/5, 0, 0, 2/5)] [(1/5, 0, 1/5, 3/5, 0), (0, 4/5, 0, 0, 1/5)] [(1/5, 0, 2/5, 2/5, 0), (0, 1/5, 0, 0, 4/5)] 
[(1/5, 0, 2/5, 2/5, 0), (0, 2/5, 0, 0, 3/5)] [(1/5, 0, 2/5, 2/5, 0), (0, 3/5, 0, 0, 2/5)] [(1/5, 0, 2/5, 2/5, 0), (0, 4/5, 0, 0, 1/5)] [(1/5, 0, 3/5, 1/5, 0), (0, 1/5, 0, 0, 4/5)] 
[(1/5, 0, 3/5, 1/5, 0), (0, 2/5, 0, 0, 3/5)] [(1/5, 0, 3/5, 1/5, 0), (0, 3/5, 0, 0, 2/5)] [(1/5, 0, 3/5, 1/5, 0), (0, 4/5, 0, 0, 1/5)] [(2/5, 0, 1/5, 2/5, 0), (0, 1/5, 0, 0, 4/5)] 
[(2/5, 0, 1/5, 2/5, 0), (0, 2/5, 0, 0, 3/5)] [(2/5, 0, 1/5, 2/5, 0), (0, 3/5, 0, 0, 2/5)] [(2/5, 0, 1/5, 2/5, 0), (0, 4/5, 0, 0, 1/5)] [(2/5, 0, 2/5, 1/5, 0), (0, 1/5, 0, 0, 4/5)] 
[(2/5, 0, 2/5, 1/5, 0), (0, 2/5, 0, 0, 3/5)] [(2/5, 0, 2/5, 1/5, 0), (0, 3/5, 0, 0, 2/5)] [(2/5, 0, 2/5, 1/5, 0), (0, 4/5, 0, 0, 1/5)] [(3/5, 0, 1/5, 1/5, 0), (0, 1/5, 0, 0, 4/5)] 
[(3/5, 0, 1/5, 1/5, 0), (0, 2/5, 0, 0, 3/5)] [(3/5, 0, 1/5, 1/5, 0), (0, 3/5, 0, 0, 2/5)] [(3/5, 0, 1/5, 1/5, 0), (0, 4/5, 0, 0, 1/5)].

The Euler characteristic computed by Theorem \ref{main} is $\chi (\widetilde{W_0/G_1})=-88$ and the Vafa formula gives the same value.

\item 
For $G_2=Z_5\times Z_5$ with generators [(0, 1/5, 4/5, 3/5, 2/5), (1/5, 1/5, 1/5, 1/5, 1/5)][(0, 1/5, 4/5, 3/5, 2/5), (1/5, 1/5, 1/5, 1/5, 1/5)],
the dual group has generators [(0, 0, 1/5, 3/5, 1/5), (0, 1/5, 0, 1/5, 3/5), (1/5, 1/5, 1/5, 1/5, 1/5)], and 
we have $21$ $G_2$-invariant deformations

(1, 1, 1, 1, 1) 
(0, 0, 1, 3, 1) (0, 0, 2, 1, 2) (0, 1, 0, 1, 3) (0, 1, 2, 2, 0) 
(0, 1, 3, 0, 1) (0, 2, 0, 2, 1) (0, 2, 1, 0, 2) (0, 3, 1, 1, 0) 
(1, 0, 0, 2, 2) (1, 0, 1, 0, 3) (1, 0, 3, 1, 0) (1, 1, 0, 3, 0) 
(1, 2, 2, 0, 0) (1, 3, 0, 0, 1) (2, 0, 1, 2, 0) (2, 0, 2, 0, 1) 
(2, 1, 0, 0, 2) (2, 2, 0, 1, 0) (3, 0, 0, 1, 1) (3, 1, 1, 0, 0);

one $G^T_2$-invariant deformation 
(1, 1, 1, 1, 1), and no Roan's pairs.

Theorem \ref{main} says that $\chi (\widetilde{W_0/G_2})=-40$ and the same value due to the Vafa formula.




\item 
For $G_3=Z_5\times Z_5\times Z_5$ with generators
 [(0, 0, 0, 1/5, 4/5), (0, 1/5, 4/5, 0, 0), (1/5, 1/5, 1/5, 1/5, 1/5)],
and $G^T_3$ with generators [(0, 1/5, 1/5, 4/5, 4/5), (1/5, 1/5, 1/5, 1/5, 1/5)],
we have 
$5$ $G_3$-invariant deformations
(1, 1, 1, 1, 1) 
(1, 0, 0, 2, 2) (1, 2, 2, 0, 0) (3, 0, 0, 1, 1) (3, 1, 1, 0, 0);
$17$ $G_3^T$-invariant deformations

(1, 1, 1, 1, 1) 
(0, 0, 0, 2, 3) (0, 0, 0, 3, 2) (0, 2, 3, 0, 0) (0, 3, 2, 0, 0) 
(1, 0, 2, 0, 2) (1, 0, 2, 1, 1) (1, 0, 2, 2, 0) (1, 1, 1, 0, 2) 
(1, 1, 1, 2, 0) (1, 2, 0, 0, 2) (1, 2, 0, 1, 1) (1, 2, 0, 2, 0) 
(3, 0, 1, 0, 1) (3, 0, 1, 1, 0) (3, 1, 0, 0, 1) (3, 1, 0, 1, 0) ;

no $(G,G^T)$-Roan's pairs and $16$ $(G^T,G)$-Roan's pairs

[(1/5, 0, 0, 2/5, 2/5), (0, 1/5, 4/5, 0, 0)] 
[(1/5, 0, 0, 2/5, 2/5), (0, 2/5, 3/5, 0, 0)] [(1/5, 0, 0, 2/5, 2/5), (0, 3/5, 2/5, 0, 0)] [(1/5, 0, 0, 2/5, 2/5), (0, 4/5, 1/5, 0, 0)] [(1/5, 2/5, 2/5, 0, 0), (0, 0, 0, 1/5, 4/5)] 
[(1/5, 2/5, 2/5, 0, 0), (0, 0, 0, 2/5, 3/5)] [(1/5, 2/5, 2/5, 0, 0), (0, 0, 0, 3/5, 2/5)] [(1/5, 2/5, 2/5, 0, 0), (0, 0, 0, 4/5, 1/5)] [(3/5, 0, 0, 1/5, 1/5), (0, 1/5, 4/5, 0, 0)] 
[(3/5, 0, 0, 1/5, 1/5), (0, 2/5, 3/5, 0, 0)] [(3/5, 0, 0, 1/5, 1/5), (0, 3/5, 2/5, 0, 0)] [(3/5, 0, 0, 1/5, 1/5), (0, 4/5, 1/5, 0, 0)] [(3/5, 1/5, 1/5, 0, 0), (0, 0, 0, 1/5, 4/5)] 
[(3/5, 1/5, 1/5, 0, 0), (0, 0, 0, 2/5, 3/5)] [(3/5, 1/5, 1/5, 0, 0), (0, 0, 0, 3/5, 2/5)] [(3/5, 1/5, 1/5, 0, 0), (0, 0, 0, 4/5, 1/5)];

$\chi(\widetilde{W_0/G_3)}=-8$.

\end{enumerate}

\begin{rem}
Interestingly, the full set of G-invariant deformations and Roan's pairs in this case corresponds to elements of the pairs of Batyrev mirror lattices $M$ and $N$, and is also the Borisov cohomology that forms the $N=2$-Virasoro vertex algebra (see, \cite{MBG}).

\end{rem}

\section{Borcea-Voisin mirror pairs from the product of K3 surfaces and elliptic curve and Fermat orbifolds}
\subsection{}
Borcea \cite{Borcea} and Voisin \cite{Voisin} independently constructed mirror families of Calabi-Yau threefolds from products of K3 surfaces and elliptic curves.

Let us recall their results.  Consider a Calabi-Yau K3 surface $S$ with an involution $i:S\to S$ and an elliptic curve $E$ with an involution $j:E\to E$.

Specifically, $i$ acts on $H^{2,0}(S)$ as a multiplication by $-1$. Let $C_1, \cdots, C_N$ be disjoin curves which 
constitute the set of fixed points of $i$. Let $g_1, \ldots, g_N$ be genera of these curves, $N'=\sum_{l=1}^N g_l$.

Four points $p_1, \cdots, p_4$ are fixed point of $j:E\to E$ with the factor $\mathbb P_1$.

These involutions define an involution $k: S\times E\to S\times E$ as the product $k(s, e)=(i(s), j(e))$. The set of fixed points of $k$ is $4N$ curves $(C_l, p_m)$, $l=1, \cdots , N$, $m=1, 2, 3, 4$. After making blow-ups along these curves, we get a smooth Calabi-Yau threefold $\widetilde{S\times E}$ with the exceptional divisor $D$. Obviously, we can lift the involution $k$ to an involution $\tilde k:\widetilde{S\times E}\to \widetilde{S\times E}$.
The following theorem is due to Voisin \cite{Voisin}:
\begin{thm}
The quotient $V=\widetilde{S\times E}/\tilde k$ is smooth and there holds
\[
h^{1,1}(V)=11+5N-N', \quad h^{2,1}(V)=11+5N'-N.
\]
\end{thm}

For such a pair $(S, E)$, C. Voisin and C. Borcea constructed another pair of elliptic curve and K3 surface for which $N$ and $N'$ swapped and consider such a pair  as a mirror pair.

By the classification theorem of Nikulin \cite{Nikulin}, the isomorphic class of a $K3$ surface with involution $\iota$ acting by $-1$ on $H^{2,0}(S)$
is determined by a triplet $(r,a,\delta)$. 
Borcea \cite{Borcea} related such a  triplet $(r,a,\delta)$ with Hodge numbers of Calabi-Yau threefold 
\[
\widetilde{S\times E}/k.
\]
Namely, the following proposition is due to Borcea \cite{Borcea}.

\begin{prop}\label{Bor1} For a K3 surface $S$ with involution $\iota$ and triplet $(r,a,\delta)$,  
   the Hodge  numbers of Calabi-Yau threefold $\widetilde{S\times E}/k
   $ are 
   \[
   h^{1,1}=5+3r-2a,
   \]
   \[
   h^{2,1}=65-3r-2a.
   \]
\end{prop}

For a triplet $(r,a,\delta)$, we have 
the number of components $N=1+\frac{r-a}{2}$ and the genera $N'=11-\frac{r+a}{2}$, and the Euler number $\chi=12r-120$.

The mirror to $E\times S$ with triple $(r,a,\delta)$ is defined by   K3 $S^{\vee}$ with triple $(20-r,a,\delta)$ due to Borcea \cite{Borcea} and Voisin \cite{Voisin}.

A goal of this section is, for Fermat threefolds related to the Borcea-Voisin construction, find a corresponding  orbifold. 
There are ten such Fermat threefolds with exponents and corresponding pairs $(r,a)$ listed below (see Tables 1-3 in \cite{GLY})
\begin{itemize}
    \item  For $K3$ $x_0^2=w^6+y^6+z^6$ with  $r=1$, $a=1$; we consider Fermat of type $(4,4,6,6,6)$ by a group isomorphic to $Z_{12}\times Z_2$, and 
    BV-mirror $z_1^4+z_2^4+z_3^5+z_3z_4^5 +z_4z_5^6$ which is not of Fermat type;
    \item For $x_0^2=w^4+y^8+z^8$  with $r=2$, $a=2$; we consider Fermat orbifold of type $(4,4,4,8,8)$ by a group isomorphic to $Z_8\times Z_2$,
    BV-mirror $z_1^4+z_2^4+z_3^4z_4+z_3z_5^6+z_3z_5^6+z_4^3z_5^4$;
    \item For $x_0^2=w^3+y^7+z^{42}$  with
     $(4,4,3,7,42)$, $r=10$, $a=0$, which is self dual due to BV-construction; we consider  Fermat orbifold of type $(3,4,4,7,42)$ by a group isomorphic to $Z_{42}\times Z_2$,
    \item  For $x_0^2=w^3+y^{10}+z^{15}$  with
     $(4,4,3,10,15)$, $r=10$, $a=4$, which is self dual due to BV-construction; we consider  Fermat orbifold of type $(3,4,4,10,20)$ by a group isomorphic to $Z_{60}\times Z_2$,
    \item  
      For $x_0^2=w^4+y^{5}+z^{20}$  with
      $r=6$, $a=4$, with mirror $z_1^4+z_2^4+z_3^4z_5+z_3z_4^4 +z_5^{13}$ due  to BV-construction; we consider  Fermat orbifold of type $(4,4,4,5,20)$ by a group isomorphic to $Z_{20}\times Z_2$,

\item For $x_0^2=w^4+y^{6}+z^{12}$  with
      $r=4$, $a=4$, 
      we consider  Fermat orbifold of type $(4,4,4,6,12)$ by a group isomorphic to $Z_{24}\times Z_2$,
    \item  For $x_0^2=w^3+y^{8}+z^{24}$  with
      $r=6$, $a=2$,  mirror $z_1^4+z_2^4+z_3^3+z_5z_4^9 +z_5^{16}$ due  to BV-construction; 
      we consider  Fermat orbifold of type $(3,4,4,8,24)$ by a group isomorphic to $Z_{24}\times Z_2$,
\item For $x_0^2=w^3+y^{12}+z^{12}$  with
      $r=2$, $a=0$,  mirror  $z_1^4+z_2^4+z_3^3z_4+z_5z_4^8 +z_5^{11}$
      due  to BV-construction; 
      we consider  Fermat orbifold of type $(3,4,4,12,12)$ by a group isomorphic to $Z_{12}\times Z_2$,

\item For $x_0^2=w^5+y^{5}+z^{10}$  with
      $r=6$, $a=4$,  mirror $z_1^4+z_2^4+z_3^5+z_5z_4^5 +z_5^8$ 
      due  to BV-construction; 
      we consider  Fermat orbifold of type $(4,4,5,5,10)$ by a group isomorphic to $Z_{20}\times Z_2$,

\item 
For $x_0^2=w^3+y^{9}+z^{18}$  with
      $r=6$, $a=2$,  mirror $z_1^4+z_2^4+z_3^3+z_5z_4^8+z_5^{21}$
      due  to BV-construction; 
      we consider  Fermat orbifold of type $(3,4,4,9,18)$ by a group isomorphic to $Z_{36}\times Z_2$.

\end{itemize}

 
\begin{thm}\label{BV-BHK}
   For any  of above ten types of Fermat  there exists a group $G$ isomorphic to  $G_0\times Z_2\subset G^{\max}$ such that the Roan's Hodge numbers coincide with the Hodge numbers 
   computed as in Proposition \ref{Bor1}.
\end{thm}
\begin{rem}
We consider BV-mirrors and corresponding BHK orbifold in another paper, since BV-mirrors are not of Fermat type.
\end{rem}


Another approach to relate Borcea-Voisin construction and Berglund-H\"ubsch-Chiodo-Ruan transposition rule is in \cite{ABS}. 
\section{Proof of Theorem \ref{BV-BHK}}
We provide the corresponding group $G$ case-by-case.
\subsection{}
Let us consider the Fermat threefold of type $(4,4,6,6,6)$ with $r=1$, $a=1$, $h^{1,1}=6$, $h^{2,1}=60$, $\chi=-108$,
\[
W_{4,4,6,6,6}=z_1^4+z_2^4+z_3^6+z_4^6+z_5^6.
\]
This threefold is of degree $12$ of the weighted projective space $\mathbb P(3,3,2,2,2)$. 

In order to relate this Fermat threefold with the quotient of the product of the elliptic curve $E=z_1^4+z_2^4+w^2$ in $\mathbb P(1,1,2)$ and the K3 surface $S=-w^2+z_3^6+z_4^6+z_5^6$ in $\mathbb P(3,1,1,1)$ by the involution, we consider an orbifold with the quotient group $Z_{12}\times Z_2$.

The involution $i$ has the only one fixed curve $0=z_3^6+z_4^6+z_5^6$ of genus $10$. Therefore $(1, 10)$ is the corresponding pair $(N,N')$.

By Theorem \ref{main}, for 
a group $G=Z_{12}\times Z_2$ with generators  
$$[(0, 1/2, 1/6, 1/6, 1/6), (1/4, 1/4, 1/6, 1/6, 1/6)],$$ the Hodge numbers the desingularization of $W_{4,4,6,6,6}/G$ 
are 
$h^{2,1}=60$ and $h^{1,1}=6$. 

Namely, for 
a group $G=Z_{12}\times Z_2$ with generators  
$$[(0, 1/2, 1/6, 1/6, 1/6), (1/4, 1/4, 1/6, 1/6, 1/6)],$$
the mirror group $G^T$ has generators 
$$[(0, 0, 0, 1/6, 5/6), (0, 0, 1/6, 0, 5/6), (1/4, 1/4, 1/6, 1/6, 1/6)].$$

There are $30$  $G$-invariant deformations 

(1, 1, 1, 1, 1) 
(0, 0, 0, 2, 4) (0, 0, 0, 3, 3) (0, 0, 0, 4, 2) (0, 0, 1, 1, 4) 
(0, 0, 1, 2, 3) (0, 0, 1, 3, 2) (0, 0, 1, 4, 1) (0, 0, 2, 0, 4) 
(0, 0, 2, 1, 3) (0, 0, 2, 2, 2) (0, 0, 2, 3, 1) (0, 0, 2, 4, 0) 
(0, 0, 3, 0, 3) (0, 0, 3, 1, 2) (0, 0, 3, 2, 1) (0, 0, 3, 3, 0) 
(0, 0, 4, 0, 2) (0, 0, 4, 1, 1) (0, 0, 4, 2, 0) (1, 1, 0, 0, 3) 
(1, 1, 0, 1, 2) (1, 1, 0, 2, 1) (1, 1, 0, 3, 0) (1, 1, 1, 0, 2) 
(1, 1, 1, 2, 0) (1, 1, 2, 0, 1) (1, 1, 2, 1, 0) (1, 1, 3, 0, 0) 
(2, 2, 0, 0, 0) .

There are $5$ $G^T$-invariant deformation
$$(1, 1, 1, 1, 1) 
(0, 0, 2, 2, 2) (0, 2, 1, 1, 1) (2, 0, 1, 1, 1) (2, 2, 0, 0, 0).$$

There is one $(G, G^T)$-Roan pair
$$
[(0, 0, 1/3, 1/3, 1/3), (1/2, 1/2, 0, 0, 0)], $$

There are $30$ $(G^T, G)$-Roan's pairs 

[(0, 0, 1/6, 1/6, 2/3), (1/4, 3/4, 0, 0, 0)] 
[(0, 0, 1/6, 1/6, 2/3), (1/2, 1/2, 0, 0, 0)] [(0, 0, 1/6, 1/6, 2/3), (3/4, 1/4, 0, 0, 0)] [(0, 0, 1/6, 1/3, 1/2), (1/4, 3/4, 0, 0, 0)] [(0, 0, 1/6, 1/3, 1/2), (1/2, 1/2, 0, 0, 0)] 
[(0, 0, 1/6, 1/3, 1/2), (3/4, 1/4, 0, 0, 0)] [(0, 0, 1/6, 1/2, 1/3), (1/4, 3/4, 0, 0, 0)] [(0, 0, 1/6, 1/2, 1/3), (1/2, 1/2, 0, 0, 0)] [(0, 0, 1/6, 1/2, 1/3), (3/4, 1/4, 0, 0, 0)] 
[(0, 0, 1/6, 2/3, 1/6), (1/4, 3/4, 0, 0, 0)] [(0, 0, 1/6, 2/3, 1/6), (1/2, 1/2, 0, 0, 0)] [(0, 0, 1/6, 2/3, 1/6), (3/4, 1/4, 0, 0, 0)] [(0, 0, 1/3, 1/6, 1/2), (1/4, 3/4, 0, 0, 0)] 
[(0, 0, 1/3, 1/6, 1/2), (1/2, 1/2, 0, 0, 0)] [(0, 0, 1/3, 1/6, 1/2), (3/4, 1/4, 0, 0, 0)] [(0, 0, 1/3, 1/3, 1/3), (1/4, 3/4, 0, 0, 0)] [(0, 0, 1/3, 1/3, 1/3), (1/2, 1/2, 0, 0, 0)] 
[(0, 0, 1/3, 1/3, 1/3), (3/4, 1/4, 0, 0, 0)] [(0, 0, 1/3, 1/2, 1/6), (1/4, 3/4, 0, 0, 0)] [(0, 0, 1/3, 1/2, 1/6), (1/2, 1/2, 0, 0, 0)] [(0, 0, 1/3, 1/2, 1/6), (3/4, 1/4, 0, 0, 0)] 
[(0, 0, 1/2, 1/6, 1/3), (1/4, 3/4, 0, 0, 0)] [(0, 0, 1/2, 1/6, 1/3), (1/2, 1/2, 0, 0, 0)] [(0, 0, 1/2, 1/6, 1/3), (3/4, 1/4, 0, 0, 0)] [(0, 0, 1/2, 1/3, 1/6), (1/4, 3/4, 0, 0, 0)] 
[(0, 0, 1/2, 1/3, 1/6), (1/2, 1/2, 0, 0, 0)] [(0, 0, 1/2, 1/3, 1/6), (3/4, 1/4, 0, 0, 0)] [(0, 0, 2/3, 1/6, 1/6), (1/4, 3/4, 0, 0, 0)] [(0, 0, 2/3, 1/6, 1/6), (1/2, 1/2, 0, 0, 0)] 
[(0, 0, 2/3, 1/6, 1/6), (3/4, 1/4, 0, 0, 0)] 


Note that   following Voisin and Borcea, we have $h^{1,1}=11+5-10=5+1$, and following orbifold construction there is 5 deformation wrt $G^T$ and one Roan pair.

Borcea \cite{Borcea} considered  the hypersurface
\[
W'=z_1^4+z_2^4+z_3^5+z_3z_4^5+z_4z_5^6
\]
in $\mathbb P(25, 25, 20, 16, 14)$ as a mirror to $\widetilde{W_{4,4,6,6,6}/G}$.

Thus,   for $\widetilde{W_{4,4,6,6,6}/G}$, we have two mirrors $\widetilde{W_{4,4,6,6,6}/G^T}$ and $W'$ or they are birational.

Note that there are 6 more subgroups in $G^{\max}$ isomorphic $Z_{12}\times Z_2$ with Euler numbers $\pm 24$ and $\pm 60$.

\subsection{}
The Fermat threefold of  the type $(4,4,4,8,8)$ in $\mathbb P(2,2,2,1,1)$, related to the Borcea-Voisin construction with $r=2$, $a=2$, $h^{1,1}=7$, $h^{2,1}=55$, $\chi=-96$,
\[
W_{4,4,4,8,8}=z_1^4+z_2^4+z_3^4+z_4^8+z_5^8.
\]
We relate the desingularization of an orbifold  $W_{4,4,4,8,8}/Z_8\times Z_2$ with the quotient of the product of the elliptic curve $E=z_1^4+z_2^4+w^2$ in $\mathbb P(1,1,2)$ and the K3 surface $S=-w^2+z_3^4+z_4^8+z_5^8$ in $\mathbb P(4,2,1,1)$ by involution $i$ as above. 
The involution $i$ has the only one fixed curve $0=z_3^4+z_4^8+z_5^8$ of genus $9$. Therefore here $(9, 1)$ is the corresponding pair $(N',N)$, Hodge numbers are
$h^{2,1}(\widetilde{E\times S/i})=55$ and $h^{1,1}(\widetilde{E\times S/i})=7$, and the Euler characteristic $\chi(W/i)=-96$.

For the group $G'=Z_8\times Z_2\subset G^{\max}$ with generators  [(0, 1/2, 1/2, 0, 0), (1/4, 1/4, 1/4, 1/8, 1/8)],  we have the Hodge numbers $h^{2,1}(\widetilde{W_{4,4,4,8,8}/G'})=55$ and $h^{1,1}(\widetilde{W_{4,4,4,8,8}/G'})=7$.

Namely, the mirror group $G'^T$ has generators  [(0, 0, 0, 1/8, 7/8), (0, 0, 1/2, 0, 1/2), (0, 1/4, 1/4, 0, 1/2), (1/4, 1/4, 1/4, 1/8, 1/8)],
there are $45$ $G'$-invariant deformations 

(1, 1, 1, 1, 1) 
(0, 0, 0, 2, 6) (0, 0, 0, 3, 5) (0, 0, 0, 4, 4) (0, 0, 0, 5, 3) 
(0, 0, 0, 6, 2) (0, 0, 2, 0, 4) (0, 0, 2, 1, 3) (0, 0, 2, 2, 2) 
(0, 0, 2, 3, 1) (0, 0, 2, 4, 0) (0, 1, 1, 0, 4) (0, 1, 1, 1, 3) 
(0, 1, 1, 2, 2) (0, 1, 1, 3, 1) (0, 1, 1, 4, 0) (0, 2, 0, 0, 4) 
(0, 2, 0, 1, 3) (0, 2, 0, 2, 2) (0, 2, 0, 3, 1) (0, 2, 0, 4, 0) 
(0, 2, 2, 0, 0) (1, 0, 0, 0, 6) (1, 0, 0, 1, 5) (1, 0, 0, 2, 4) 
(1, 0, 0, 3, 3) (1, 0, 0, 4, 2) (1, 0, 0, 5, 1) (1, 0, 0, 6, 0) 
(1, 0, 2, 0, 2) (1, 0, 2, 1, 1) (1, 0, 2, 2, 0) (1, 1, 1, 0, 2) 
(1, 1, 1, 2, 0) (1, 2, 0, 0, 2) (1, 2, 0, 1, 1) (1, 2, 0, 2, 0) 
(2, 0, 0, 0, 4) (2, 0, 0, 1, 3) (2, 0, 0, 2, 2) (2, 0, 0, 3, 1) 
(2, 0, 0, 4, 0) (2, 0, 2, 0, 0) (2, 1, 1, 0, 0) (2, 2, 0, 0, 0);

and four $G'^T$-invariant deformations 
(1, 1, 1, 1, 1) 
(0, 0, 0, 4, 4) (0, 2, 2, 0, 0) (2, 0, 0, 2, 2);

There are three $(G',G'^T)$-Roan's pairs [(1/2, 0, 0, 1/4, 1/4), (0, 1/4, 3/4, 0, 0)] 
[(1/2, 0, 0, 1/4, 1/4), (0, 1/2, 1/2, 0, 0)] [(1/2, 0, 0, 1/4, 1/4), (0, 3/4, 1/4, 0, 0)];
and and $10$ $(G'^T,G')$-Roan's pairs 

[(1/4, 0, 0, 1/8, 5/8), (0, 1/2, 1/2, 0, 0)] 
[(1/4, 0, 0, 1/4, 1/2), (0, 1/2, 1/2, 0, 0)] [(1/4, 0, 0, 3/8, 3/8), (0, 1/2, 1/2, 0, 0)] [(1/4, 0, 0, 1/2, 1/4), (0, 1/2, 1/2, 0, 0)] [(1/4, 0, 0, 5/8, 1/8), (0, 1/2, 1/2, 0, 0)] 
[(1/2, 0, 0, 1/8, 3/8), (0, 1/2, 1/2, 0, 0)] [(1/2, 0, 0, 1/4, 1/4), (0, 1/2, 1/2, 0, 0)] [(1/2, 0, 0, 3/8, 1/8), (0, 1/2, 1/2, 0, 0)] [(1/2, 1/4, 1/4, 0, 0), (0, 0, 0, 1/2, 1/2)] 
[(3/4, 0, 0, 1/8, 1/8), (0, 1/2, 1/2, 0, 0)]. 
[(0, 0, 1/2, 0, 1/2), (1/3, 1/4, 1/4, 1/7, 1/42)]

Thus, $h^{2,1}(\widetilde{W_{4,4,4,8,8}/G'})=55$ and $h^{1,1}(\widetilde{W_{4,4,4,8,8}/G'})=7$, and $\chi(\widetilde{W_{4,4,4,8,8}/G'})=-96$.

In this example, we have $h^{2,1}(\widetilde{W_{4,4,4,8,8}/G})=55$ is constituted of $45$ deformations and $10$ Roan's pairs.

\subsection{}
Here we consider the Fermat threefold of  the type $(3, 4, 4, 7, 42)$ in $\mathbb P(28, 21, 21, 12, 2)$, related to the B-V construction of the threefold of 
the product the elliptic curve $E$ $z_1^4+z_2^4=z_0^2$ and $K3$ surface $S$ $z_0^2=z_3^3+z_4^7+z_5^{42}$, with $r=10$, $a=0$, $h^{1,1}=35$, $h^{2,1}=35$,
$\chi=0$ and sef dual due to BV-construction,
\[
W_{3,4,4,7,42}=z_1^3+z_2^4+z_3^4+z_4^7+z_5^{42}.
\]
We consider the orbifold 
\[
\widetilde{W_{(3,4,4,7,42)}/Z_{84}\times Z_2},
\]
Note that $G^{\max}=G_0\times Z_2$, so the mirror orbifold is with the group $G_0$.

We get the generators $g_1=(0, 0, 1/2, 0, 1/2)$ and $g_0=(1/3, 1/4, 1/4, 1/7, 1/42)$ for $G_0\times Z_2=G^{\max}$
and the mirror $G_0$ has  the generator $g_0=(1/3, 1/4, 1/4, 1/7, 1/42)$.

There are seventeen  $G^{\max}$-invariant deformations

(1, 1, 1, 1, 1) 
(0, 0, 0, 1, 36) (0, 0, 0, 2, 30) (0, 0, 0, 3, 24) (0, 0, 0, 4, 18) 
(0, 0, 0, 5, 12) (0, 1, 1, 0, 21) (0, 1, 1, 1, 15) (0, 1, 1, 2, 9) 
(0, 1, 1, 3, 3) (0, 2, 2, 0, 0) (1, 0, 0, 0, 28) (1, 0, 0, 1, 22) 
(1, 0, 0, 2, 16) (1, 0, 0, 3, 10) (1, 0, 0, 4, 4) (1, 1, 1, 0, 7),

and twenty nine $G_0$-invariant deformations

(1, 1, 1, 1, 1) 
(0, 0, 0, 1, 36) (0, 0, 0, 2, 30) (0, 0, 0, 3, 24) (0, 0, 0, 4, 18) 
(0, 0, 0, 5, 12) (0, 0, 2, 0, 21) (0, 0, 2, 1, 15) (0, 0, 2, 2, 9) 
(0, 0, 2, 3, 3) (0, 1, 1, 0, 21) (0, 1, 1, 1, 15) (0, 1, 1, 2, 9) 
(0, 1, 1, 3, 3) (0, 2, 0, 0, 21) (0, 2, 0, 1, 15) (0, 2, 0, 2, 9) 
(0, 2, 0, 3, 3) (0, 2, 2, 0, 0) (1, 0, 0, 0, 28) (1, 0, 0, 1, 22) 
(1, 0, 0, 2, 16) (1, 0, 0, 3, 10) (1, 0, 0, 4, 4) (1, 0, 2, 0, 7) 
(1, 0, 2, 1, 1) (1, 1, 1, 0, 7) (1, 2, 0, 0, 7) (1, 2, 0, 1, 1).

There are $6$ $(G_0,G^{\max})$ Roan's pairs

[(1/3, 0, 0, 1/7, 11/21), (0, 1/2, 1/2, 0, 0)] 
[(1/3, 0, 0, 2/7, 8/21), (0, 1/2, 1/2, 0, 0)] [(1/3, 0, 0, 3/7, 5/21), (0, 1/2, 1/2, 0, 0)] [(1/3, 0, 0, 4/7, 2/21), (0, 1/2, 1/2, 0, 0)] [(2/3, 0, 0, 1/7, 4/21), (0, 1/2, 1/2, 0, 0)] 
[(2/3, 0, 0, 2/7, 1/21), (0, 1/2, 1/2, 0, 0)],

and there are are $18$ $(G^{\max},G_0)$ Roan's pairs

[(1/3, 0, 0, 1/7, 11/21), (0, 1/4, 3/4, 0, 0)] 
[(1/3, 0, 0, 1/7, 11/21), (0, 1/2, 1/2, 0, 0)] [(1/3, 0, 0, 1/7, 11/21), (0, 3/4, 1/4, 0, 0)] [(1/3, 0, 0, 2/7, 8/21), (0, 1/4, 3/4, 0, 0)] [(1/3, 0, 0, 2/7, 8/21), (0, 1/2, 1/2, 0, 0)] 
[(1/3, 0, 0, 2/7, 8/21), (0, 3/4, 1/4, 0, 0)] [(1/3, 0, 0, 3/7, 5/21), (0, 1/4, 3/4, 0, 0)] [(1/3, 0, 0, 3/7, 5/21), (0, 1/2, 1/2, 0, 0)] [(1/3, 0, 0, 3/7, 5/21), (0, 3/4, 1/4, 0, 0)] 
[(1/3, 0, 0, 4/7, 2/21), (0, 1/4, 3/4, 0, 0)] [(1/3, 0, 0, 4/7, 2/21), (0, 1/2, 1/2, 0, 0)] [(1/3, 0, 0, 4/7, 2/21), (0, 3/4, 1/4, 0, 0)] [(2/3, 0, 0, 1/7, 4/21), (0, 1/4, 3/4, 0, 0)] 
[(2/3, 0, 0, 1/7, 4/21), (0, 1/2, 1/2, 0, 0)] [(2/3, 0, 0, 1/7, 4/21), (0, 3/4, 1/4, 0, 0)] [(2/3, 0, 0, 2/7, 1/21), (0, 1/4, 3/4, 0, 0)] [(2/3, 0, 0, 2/7, 1/21), (0, 1/2, 1/2, 0, 0)] 
[(2/3, 0, 0, 2/7, 1/21), (0, 3/4, 1/4, 0, 0)].

Therefore, for the Euler characteristics we have 
\[
\chi(\widetilde{W_{(3,4,4,7,42)}/G^{\max}}=\chi(\widetilde{W_{(3,4,4,7,42)}/G_0}=0.
\]

However, we can see that the deformations and Roans pairs are different for $\widetilde{W_{(3,4,4,7,42)}/G^{\max}}$ and for the mirror $\widetilde{W_{(3,4,4,7,42)}/G_0}$
despite the Hodge numbers $h^{2,1}(\widetilde{W_{(3,4,4,7,42)}/G^{\max}})=35$ and $h^{2,1}(\widetilde{W_{(3,4,4,7,42)}/G_0})=35$ are the same.
Therefore we get a multiple mirror wrt self mirror due to the BV-construction.

\subsection{}
For  Fermat with exponents  $(4,4,3,10,15)$,  $r=10$, $a=4$; $h^{1,1}=27$ and $h^{2,1}=27$, $\chi=0$,
\[
W_{4,4,3,10,15}=z_1^3+z_2^4+z_3^4+z_4^{10}+z_5^{15}
\]
quotient by the group $G^{\max}:=G_0\times Z_2\cong Z_{60}\times Z_2$ and the mirror groups $G^T=G_0$, we get
thirteen  $G^{\max}$-invariant deformations,
(1, 1, 1, 1, 1) 
(0, 0, 0, 2, 12) (0, 0, 0, 4, 9) (0, 0, 0, 6, 6) (0, 0, 0, 8, 3) 
(0, 1, 1, 1, 6) (0, 1, 1, 3, 3) (0, 1, 1, 5, 0) (0, 2, 2, 0, 0) 
(1, 0, 0, 0, 10) (1, 0, 0, 2, 7) (1, 0, 0, 4, 4) (1, 0, 0, 6, 1);

twenty-one $G_0$-invariant deformations, 
(1, 1, 1, 1, 1) 
(0, 0, 0, 2, 12) (0, 0, 0, 4, 9) (0, 0, 0, 6, 6) (0, 0, 0, 8, 3) 
(0, 0, 2, 1, 6) (0, 0, 2, 3, 3) (0, 0, 2, 5, 0) (0, 1, 1, 1, 6) 
(0, 1, 1, 3, 3) (0, 1, 1, 5, 0) (0, 2, 0, 1, 6) (0, 2, 0, 3, 3) 
(0, 2, 0, 5, 0) (0, 2, 2, 0, 0) (1, 0, 0, 0, 10) (1, 0, 0, 2, 7) 
(1, 0, 0, 4, 4) (1, 0, 0, 6, 1) (1, 0, 2, 1, 1) (1, 2, 0, 1, 1);

 six $(G_0,G^{\max)}$- Roan's pairs, 
[(0, 1/4, 1/4, 1/2, 0), (1/3, 0, 0, 0, 2/3)] 
[(0, 1/4, 1/4, 1/2, 0), (2/3, 0, 0, 0, 1/3)] [(1/3, 0, 0, 1/5, 7/15), (0, 1/2, 1/2, 0, 0)] [(1/3, 0, 0, 2/5, 4/15), (0, 1/2, 1/2, 0, 0)] [(1/3, 0, 0, 3/5, 1/15), (0, 1/2, 1/2, 0, 0)] 
[(2/3, 0, 0, 1/5, 2/15), (0, 1/2, 1/2, 0, 0)] ;

and fourteen $(G^{\max},G_0)$- Roan's pairs,
[(0, 1/4, 1/4, 1/2, 0), (1/3, 0, 0, 0, 2/3)] 
[(0, 1/4, 1/4, 1/2, 0), (2/3, 0, 0, 0, 1/3)] [(1/3, 0, 0, 1/5, 7/15), (0, 1/4, 3/4, 0, 0)] [(1/3, 0, 0, 1/5, 7/15), (0, 1/2, 1/2, 0, 0)] [(1/3, 0, 0, 1/5, 7/15), (0, 3/4, 1/4, 0, 0)] 
[(1/3, 0, 0, 2/5, 4/15), (0, 1/4, 3/4, 0, 0)] [(1/3, 0, 0, 2/5, 4/15), (0, 1/2, 1/2, 0, 0)] [(1/3, 0, 0, 2/5, 4/15), (0, 3/4, 1/4, 0, 0)] [(1/3, 0, 0, 3/5, 1/15), (0, 1/4, 3/4, 0, 0)] 
[(1/3, 0, 0, 3/5, 1/15), (0, 1/2, 1/2, 0, 0)] [(1/3, 0, 0, 3/5, 1/15), (0, 3/4, 1/4, 0, 0)] [(2/3, 0, 0, 1/5, 2/15), (0, 1/4, 3/4, 0, 0)] [(2/3, 0, 0, 1/5, 2/15), (0, 1/2, 1/2, 0, 0)] 
[(2/3, 0, 0, 1/5, 2/15), (0, 3/4, 1/4, 0, 0)]

Due to the BV-construction $\widetilde{E\times S/j\times \iota}$ is self dual in this case.

\subsection{} For the Fermat exponents $(4,4,4,5,20)$,  with $r=6$, $a=4$, $h^{1,1}=15$, $h^{2,1}=39$, $\chi=-48$,
$K3$ is defined by $S=z_3^4z_5+z_3z_4^4 +z_5^{13}$, and 
the orbifold of
\[
W_{4,4,4,5,20}=z_1^4+z_2^4+z_3^4+z_4^5+z_5^{20},
\]
by the group $G=G_0\times Z_2$ with generators 
 [(0, 0, 1/2, 0, 1/2), (1/4, 1/4, 1/4, 1/5, 1/20)], we have generators of the mirror group
 [(0, 0, 1/4, 0, 3/4), (0, 1/2, 0, 0, 1/2), (1/4, 1/4, 1/4, 1/5, 1/20)],

twenty nine $G$-invariant deformations 
(1, 1, 1, 1, 1) 
(0, 0, 0, 1, 16) (0, 0, 0, 2, 12) (0, 0, 0, 3, 8) (0, 0, 1, 0, 15) 
(0, 0, 1, 1, 11) (0, 0, 1, 2, 7) (0, 0, 1, 3, 3) (0, 0, 2, 0, 10) 
(0, 0, 2, 1, 6) (0, 0, 2, 2, 2) (0, 2, 0, 0, 10) (0, 2, 0, 1, 6) 
(0, 2, 0, 2, 2) (0, 2, 1, 0, 5) (0, 2, 1, 1, 1) (0, 2, 2, 0, 0) 
(1, 1, 0, 0, 10) (1, 1, 0, 1, 6) (1, 1, 0, 2, 2) (1, 1, 1, 0, 5) 
(1, 1, 2, 0, 0) (2, 0, 0, 0, 10) (2, 0, 0, 1, 6) (2, 0, 0, 2, 2) 
(2, 0, 1, 0, 5) (2, 0, 1, 1, 1) (2, 0, 2, 0, 0) (2, 2, 0, 0, 0); 

nine $G^T$-invariant deformations, 
(1, 1, 1, 1, 1) 
(0, 0, 0, 1, 16) (0, 0, 0, 2, 12) (0, 0, 0, 3, 8) (0, 0, 2, 0, 10) 
(0, 0, 2, 1, 6) (0, 0, 2, 2, 2) (1, 1, 1, 0, 5) (2, 2, 0, 0, 0);

six $(G,G^T)$- Roan's pairs, 
[(0, 0, 1/2, 1/5, 3/10), (1/4, 3/4, 0, 0, 0)] 
[(0, 0, 1/2, 1/5, 3/10), (1/2, 1/2, 0, 0, 0)] [(0, 0, 1/2, 1/5, 3/10), (3/4, 1/4, 0, 0, 0)] [(0, 0, 1/2, 2/5, 1/10), (1/4, 3/4, 0, 0, 0)] [(0, 0, 1/2, 2/5, 1/10), (1/2, 1/2, 0, 0, 0)] 
[(0, 0, 1/2, 2/5, 1/10), (3/4, 1/4, 0, 0, 0)];

and ten $(G^T,G)$-Roan's pairs, 
[(0, 0, 1/4, 1/5, 11/20), (1/2, 1/2, 0, 0, 0)] 
[(0, 0, 1/4, 2/5, 7/20), (1/2, 1/2, 0, 0, 0)] [(0, 0, 1/4, 3/5, 3/20), (1/2, 1/2, 0, 0, 0)] [(0, 0, 1/2, 1/5, 3/10), (1/2, 1/2, 0, 0, 0)] [(0, 0, 1/2, 2/5, 1/10), (1/2, 1/2, 0, 0, 0)] 
[(0, 0, 3/4, 1/5, 1/20), (1/2, 1/2, 0, 0, 0)] [(1/4, 1/4, 1/2, 0, 0), (0, 0, 0, 1/5, 4/5)] [(1/4, 1/4, 1/2, 0, 0), (0, 0, 0, 2/5, 3/5)] [(1/4, 1/4, 1/2, 0, 0), (0, 0, 0, 3/5, 2/5)] 
[(1/4, 1/4, 1/2, 0, 0), (0, 0, 0, 4/5, 1/5)].

There are two more isomorphic orbifolds with $Z_{20}\times Z_2$.

\subsection{}
For the type $(4,4,4,6,12)$, we have $r=4$, $a=4$; 
$h^{1,1}=9$, $h^{2,1}=45$, $\chi=-72$.  Note that there is no mirror in Tables \cite{GLY}.

There are three subgroups isomorphic to $Z_{24}\times Z_2$ corresponding to BV-construction with $r=a=4$, and four subgroups with another Hodge numbers and Euler characteristic either $-60$ or $-84$.

A good subgroups is with generators 
[(0, 0, 1/2, 0, 1/2), (1/4, 1/4, 1/4, 1/6, 1/12)];
the mirror groups has the generators 
 [(0, 0, 0, 1/6, 5/6), (0, 0, 1/4, 0, 3/4), (0, 1/2, 0, 0, 1/2), (1/4, 1/4, 1/4, 1/6, 1/12)];

There are 35 $G$-invariant deformations, 
(1, 1, 1, 1, 1) 
(0, 0, 0, 1, 10) (0, 0, 0, 2, 8) (0, 0, 0, 3, 6) (0, 0, 0, 4, 4) 
(0, 0, 1, 0, 9) (0, 0, 1, 1, 7) (0, 0, 1, 2, 5) (0, 0, 1, 3, 3) 
(0, 0, 1, 4, 1) (0, 0, 2, 0, 6) (0, 0, 2, 1, 4) (0, 0, 2, 2, 2) 
(0, 0, 2, 3, 0) (0, 2, 0, 0, 6) (0, 2, 0, 1, 4) (0, 2, 0, 2, 2) 
(0, 2, 0, 3, 0) (0, 2, 1, 0, 3) (0, 2, 1, 1, 1) (0, 2, 2, 0, 0) 
(1, 1, 0, 0, 6) (1, 1, 0, 1, 4) (1, 1, 0, 2, 2) (1, 1, 0, 3, 0) 
(1, 1, 1, 0, 3) (1, 1, 2, 0, 0) (2, 0, 0, 0, 6) (2, 0, 0, 1, 4) 
(2, 0, 0, 2, 2) (2, 0, 0, 3, 0) (2, 0, 1, 0, 3) (2, 0, 1, 1, 1) 
(2, 0, 2, 0, 0) (2, 2, 0, 0, 0) ;

six $G^T$-invariant deformations,
(1, 1, 1, 1, 1) 
(0, 0, 0, 2, 8) (0, 0, 0, 4, 4) (0, 0, 2, 0, 6) (0, 0, 2, 2, 2) 
(2, 2, 0, 0, 0) ;


three $(G,G^T)$-invariant pairs, 
[(0, 0, 1/2, 1/3, 1/6), (1/4, 3/4, 0, 0, 0)] 
[(0, 0, 1/2, 1/3, 1/6), (1/2, 1/2, 0, 0, 0)] [(0, 0, 1/2, 1/3, 1/6), (3/4, 1/4, 0, 0, 0)] ;

and ten $(G,G^T)$-invariant pairs, 
[(0, 0, 1/4, 1/6, 7/12), (1/2, 1/2, 0, 0, 0)] 
[(0, 0, 1/4, 1/3, 5/12), (1/2, 1/2, 0, 0, 0)] [(0, 0, 1/4, 1/2, 1/4), (1/2, 1/2, 0, 0, 0)] [(0, 0, 1/4, 2/3, 1/12), (1/2, 1/2, 0, 0, 0)] [(0, 0, 1/2, 1/6, 1/3), (1/2, 1/2, 0, 0, 0)] 
[(0, 0, 1/2, 1/3, 1/6), (1/2, 1/2, 0, 0, 0)] [(0, 0, 3/4, 1/6, 1/12), (1/2, 1/2, 0, 0, 0)] [(1/4, 1/4, 0, 1/2, 0), (0, 0, 1/2, 0, 1/2)] [(1/4, 1/4, 1/2, 0, 0), (0, 0, 0, 1/3, 2/3)] 
[(1/4, 1/4, 1/2, 0, 0), (0, 0, 0, 2/3, 1/3)].

Thus, we have $h^{1,1}(\widetilde{W_{4,4,4,5,20}/G}=9$, $h^{2,1}(\widetilde{W_{4,4,4,5,20}/G}=45$, and dual due to BHK is a mirror.

\subsection{}
For $K3$  of type $(2,3,8,24)$ with  $r=6$, $a=2$; 
BV-mirror is $z_1^4+z_2^4+z_3^3+z_5z_4^9 +z_5^{16}$. 

For  $W=z_1^3+z_2^4+z_3^4+z_4^8 +z_5^{24}$ quotient by 
$G=G_0\times Z_2=Z_{24}\times Z_2$ with generators 
[(0, 0, 0, 1/4, 3/4), (1/3, 1/4, 1/4, 1/8, 1/24)],
and 
the mirror $G^T$ with generators  [(0, 0, 0, 1/8, 7/8), (0, 0, 1/2, 0, 1/2), (1/3, 1/4, 1/4, 1/8, 1/24)],
we get the following deformations and Roan's pairs.

There are 34 $G$-invariant deformations, 
(1, 1, 1, 1, 1) 
(0, 0, 0, 1, 21) (0, 0, 0, 2, 18) (0, 0, 0, 3, 15) (0, 0, 0, 4, 12) 
(0, 0, 0, 5, 9) (0, 0, 0, 6, 6) (0, 0, 2, 0, 12) (0, 0, 2, 1, 9) 
(0, 0, 2, 2, 6) (0, 0, 2, 3, 3) (0, 0, 2, 4, 0) (0, 1, 1, 0, 12) 
(0, 1, 1, 1, 9) (0, 1, 1, 2, 6) (0, 1, 1, 3, 3) (0, 1, 1, 4, 0) 
(0, 2, 0, 0, 12) (0, 2, 0, 1, 9) (0, 2, 0, 2, 6) (0, 2, 0, 3, 3) 
(0, 2, 0, 4, 0) (0, 2, 2, 0, 0) (1, 0, 0, 0, 16) (1, 0, 0, 1, 13) 
(1, 0, 0, 2, 10) (1, 0, 0, 3, 7) (1, 0, 0, 4, 4) (1, 0, 0, 5, 1) 
(1, 0, 2, 0, 4) (1, 0, 2, 1, 1) (1, 1, 1, 0, 4) (1, 2, 0, 0, 4) 
(1, 2, 0, 1, 1). 

There are ten $G^T$-invariant deformations, 
(1, 1, 1, 1, 1) 
(0, 0, 0, 2, 18) (0, 0, 0, 4, 12) (0, 0, 0, 6, 6) (0, 1, 1, 1, 9) 
(0, 1, 1, 3, 3) (0, 2, 2, 0, 0) (1, 0, 0, 0, 16) (1, 0, 0, 2, 10) 
(1, 0, 0, 4, 4).


There are nine $(G,G^T)$-Roans pairs, 

[(1/3, 0, 0, 1/4, 5/12), (0, 1/4, 3/4, 0, 0)] 
[(1/3, 0, 0, 1/4, 5/12), (0, 1/2, 1/2, 0, 0)] [(1/3, 0, 0, 1/4, 5/12), (0, 3/4, 1/4, 0, 0)] [(1/3, 0, 0, 1/2, 1/6), (0, 1/4, 3/4, 0, 0)] [(1/3, 0, 0, 1/2, 1/6), (0, 1/2, 1/2, 0, 0)] 
[(1/3, 0, 0, 1/2, 1/6), (0, 3/4, 1/4, 0, 0)] [(2/3, 0, 0, 1/4, 1/12), (0, 1/4, 3/4, 0, 0)] [(2/3, 0, 0, 1/4, 1/12), (0, 1/2, 1/2, 0, 0)] [(2/3, 0, 0, 1/4, 1/12), (0, 3/4, 1/4, 0, 0)] .

There are nine $(G^T,G)$-Roans pairs, 
[(0, 1/4, 1/4, 1/2, 0), (1/3, 0, 0, 0, 2/3)] 
[(0, 1/4, 1/4, 1/2, 0), (2/3, 0, 0, 0, 1/3)] [(1/3, 0, 0, 1/8, 13/24), (0, 1/2, 1/2, 0, 0)] [(1/3, 0, 0, 1/4, 5/12), (0, 1/2, 1/2, 0, 0)] [(1/3, 0, 0, 3/8, 7/24), (0, 1/2, 1/2, 0, 0)] 
[(1/3, 0, 0, 1/2, 1/6), (0, 1/2, 1/2, 0, 0)] [(1/3, 0, 0, 5/8, 1/24), (0, 1/2, 1/2, 0, 0)] [(2/3, 0, 0, 1/8, 5/24), (0, 1/2, 1/2, 0, 0)] [(2/3, 0, 0, 1/4, 1/12), (0, 1/2, 1/2, 0, 0)].

We have $h^{1,1}(\widetilde{W_{3,4,4,5,8}/G}=19=5-4+18$, $h^{2,1}(\widetilde{W_{3,4,4,5,8}/G}=43=65-4-18$.

\subsection{}
For $K3$ with $(2,3,12,12)$, $r=2$, $a=0$, and  BV-mirror $z_1^4+z_2^4+z_3^3z_4+z_5z_4^8 +z_5^{11}$, 
there are three isomorphic subgroups on the BHK orbifold side.

Namely, we can quotient $W=z_1^3+z_2^4+z_3^4+z_4^{12} +z_5^{12}$ by the group 
$G=Z_{12}\times Z_2$ with generators 
[(0, 0, 0, 1/2, 1/2), (1/3, 1/4, 1/4, 1/12, 1/12)], has the mirror $G^T$ with generators
[(0, 0, 0, 1/12, 11/12), (0, 0, 1/2, 0, 1/2), (1/3, 1/4, 1/4, 1/12, 1/12)].

Then are $49$ $G$-invariant deformations, 
(1, 1, 1, 1, 1) 
(0, 0, 0, 2, 10) (0, 0, 0, 3, 9) (0, 0, 0, 4, 8) (0, 0, 0, 5, 7) 
(0, 0, 0, 6, 6) (0, 0, 0, 7, 5) (0, 0, 0, 8, 4) (0, 0, 0, 9, 3) 
(0, 0, 0, 10, 2) (0, 0, 2, 0, 6) (0, 0, 2, 1, 5) (0, 0, 2, 2, 4) 
(0, 0, 2, 3, 3) (0, 0, 2, 4, 2) (0, 0, 2, 5, 1) (0, 0, 2, 6, 0) 
(0, 1, 1, 0, 6) (0, 1, 1, 1, 5) (0, 1, 1, 2, 4) (0, 1, 1, 3, 3) 
(0, 1, 1, 4, 2) (0, 1, 1, 5, 1) (0, 1, 1, 6, 0) (0, 2, 0, 0, 6) 
(0, 2, 0, 1, 5) (0, 2, 0, 2, 4) (0, 2, 0, 3, 3) (0, 2, 0, 4, 2) 
(0, 2, 0, 5, 1) (0, 2, 0, 6, 0) (0, 2, 2, 0, 0) (1, 0, 0, 0, 8) 
(1, 0, 0, 1, 7) (1, 0, 0, 2, 6) (1, 0, 0, 3, 5) (1, 0, 0, 4, 4) 
(1, 0, 0, 5, 3) (1, 0, 0, 6, 2) (1, 0, 0, 7, 1) (1, 0, 0, 8, 0) 
(1, 0, 2, 0, 2) (1, 0, 2, 1, 1) (1, 0, 2, 2, 0) (1, 1, 1, 0, 2) 
(1, 1, 1, 2, 0) (1, 2, 0, 0, 2) (1, 2, 0, 1, 1) (1, 2, 0, 2, 0).

There are five $G^T$-invariant deformations, 
(1, 1, 1, 1, 1) 
(0, 0, 0, 6, 6) (0, 1, 1, 3, 3) (0, 2, 2, 0, 0) (1, 0, 0, 4, 4) 


There are six $(G,G^T)$-Roans pairs, 
[(1/3, 0, 0, 1/3, 1/3), (0, 1/4, 3/4, 0, 0)] 
[(1/3, 0, 0, 1/3, 1/3), (0, 1/2, 1/2, 0, 0)] [(1/3, 0, 0, 1/3, 1/3), (0, 3/4, 1/4, 0, 0)] [(2/3, 0, 0, 1/6, 1/6), (0, 1/4, 3/4, 0, 0)] [(2/3, 0, 0, 1/6, 1/6), (0, 1/2, 1/2, 0, 0)] 
[(2/3, 0, 0, 1/6, 1/6), (0, 3/4, 1/4, 0, 0)].

There are ten $(G^T,G)$-Roans pairs, 
[(1/3, 0, 0, 1/12, 7/12), (0, 1/2, 1/2, 0, 0)] 
[(1/3, 0, 0, 1/6, 1/2), (0, 1/2, 1/2, 0, 0)] [(1/3, 0, 0, 1/4, 5/12), (0, 1/2, 1/2, 0, 0)] [(1/3, 0, 0, 1/3, 1/3), (0, 1/2, 1/2, 0, 0)] [(1/3, 0, 0, 5/12, 1/4), (0, 1/2, 1/2, 0, 0)] 
[(1/3, 0, 0, 1/2, 1/6), (0, 1/2, 1/2, 0, 0)] [(1/3, 0, 0, 7/12, 1/12), (0, 1/2, 1/2, 0, 0)] [(2/3, 0, 0, 1/12, 1/4), (0, 1/2, 1/2, 0, 0)] [(2/3, 0, 0, 1/6, 1/6), (0, 1/2, 1/2, 0, 0)] 
[(2/3, 0, 0, 1/4, 1/12), (0, 1/2, 1/2, 0, 0)].

Thus, $h^{1,1}(\widetilde{W_{3,4,4,12,12}/G}=11=5+6$, $h^{2,1}(\widetilde{W_{3,4,4,12,12}/G}=59=65-6$.

\subsection{}
 For  $K3$  with $(2,5,5,10)$, $r=6$, $a=4$, and BV-mirror $z_1^4+z_2^4+z_3^5+z_5z_4^5 +z_5^8$; there is only one subgroup of size 
$40$, and for the orbifold obtained by quotient 
$W=z_1^4+z_2^4 +z_3^5 +z_4^5 +z_5^{10}$, by the group
with the generators 
 [(0, 1/2, 0, 0, 1/2), (1/4, 1/4, 1/5, 1/5, 1/10)],
 ( the mirror groups $G^T$ has generators  [(0, 0, 0, 1/5, 4/5), (1/4, 1/4, 1/5, 1/5, 1/10)]), 
 we have the following.

There are $21$ $G$-invariant deformations, 
(1, 1, 1, 1, 1) 
(0, 0, 0, 1, 8) (0, 0, 0, 2, 6) (0, 0, 0, 3, 4) (0, 0, 1, 0, 8) 
(0, 0, 1, 1, 6) (0, 0, 1, 2, 4) (0, 0, 1, 3, 2) (0, 0, 2, 0, 6) 
(0, 0, 2, 1, 4) (0, 0, 2, 2, 2) (0, 0, 2, 3, 0) (0, 0, 3, 0, 4) 
(0, 0, 3, 1, 2) (0, 0, 3, 2, 0) (1, 1, 0, 0, 5) (1, 1, 0, 1, 3) 
(1, 1, 0, 2, 1) (1, 1, 1, 0, 3) (1, 1, 2, 0, 1) (2, 2, 0, 0, 0).  

There are $9$ $G^T$-invariant deformations, 
(1, 1, 1, 1, 1) 
(0, 0, 1, 1, 6) (0, 0, 2, 2, 2) (0, 2, 0, 0, 5) (0, 2, 1, 1, 1) 
(1, 1, 0, 0, 5) (2, 0, 0, 0, 5) (2, 0, 1, 1, 1) (2, 2, 0, 0, 0) .


There are six $(G,G^T)$-Roans pairs, 
[(0, 0, 1/5, 1/5, 3/5), (1/2, 1/2, 0, 0, 0)] 
[(0, 0, 2/5, 2/5, 1/5), (1/2, 1/2, 0, 0, 0)] [(1/4, 1/4, 0, 0, 1/2), (0, 0, 1/5, 4/5, 0)] [(1/4, 1/4, 0, 0, 1/2), (0, 0, 2/5, 3/5, 0)] [(1/4, 1/4, 0, 0, 1/2), (0, 0, 3/5, 2/5, 0)] 
[(1/4, 1/4, 0, 0, 1/2), (0, 0, 4/5, 1/5, 0)] ;

and there are $18$ $(G,G^T)$-Roans pairs, 
[(0, 0, 1/5, 1/5, 3/5), (1/4, 3/4, 0, 0, 0)] 
[(0, 0, 1/5, 1/5, 3/5), (1/2, 1/2, 0, 0, 0)] [(0, 0, 1/5, 1/5, 3/5), (3/4, 1/4, 0, 0, 0)] [(0, 0, 1/5, 2/5, 2/5), (1/4, 3/4, 0, 0, 0)] [(0, 0, 1/5, 2/5, 2/5), (1/2, 1/2, 0, 0, 0)] 
[(0, 0, 1/5, 2/5, 2/5), (3/4, 1/4, 0, 0, 0)] [(0, 0, 1/5, 3/5, 1/5), (1/4, 3/4, 0, 0, 0)] [(0, 0, 1/5, 3/5, 1/5), (1/2, 1/2, 0, 0, 0)] [(0, 0, 1/5, 3/5, 1/5), (3/4, 1/4, 0, 0, 0)] 
[(0, 0, 2/5, 1/5, 2/5), (1/4, 3/4, 0, 0, 0)] [(0, 0, 2/5, 1/5, 2/5), (1/2, 1/2, 0, 0, 0)] [(0, 0, 2/5, 1/5, 2/5), (3/4, 1/4, 0, 0, 0)] [(0, 0, 2/5, 2/5, 1/5), (1/4, 3/4, 0, 0, 0)] 
[(0, 0, 2/5, 2/5, 1/5), (1/2, 1/2, 0, 0, 0)] [(0, 0, 2/5, 2/5, 1/5), (3/4, 1/4, 0, 0, 0)] [(0, 0, 3/5, 1/5, 1/5), (1/4, 3/4, 0, 0, 0)] [(0, 0, 3/5, 1/5, 1/5), (1/2, 1/2, 0, 0, 0)] 
[(0, 0, 3/5, 1/5, 1/5), (3/4, 1/4, 0, 0, 0)].

Thus, $h^{1,1}(\widetilde{W_{3,4,4,12,12}/G})=15=5+10$, $h^{2,1}(\widetilde{W_{3,4,4,12,12}/G})=39=65-8-18$.

\subsection{} 
For 
 $K3$ surface with $(2,3,9,18)$, $r=6$, $a=2$; mirror $z_1^4+z_2^4+z_3^3+z_5z_4^8+z_5^{21}$,
 and, for Fermat threefold 
 \[W=z_1^3+z_2^4 +z_3^4 +z_4^9 +z_5^{18},\]
  there is only one subgroup $G=Z_{36}\times Z_2\subset G^{\max}$ of size 
$72$.

Here are the generators of $G$: [(0, 0, 1/2, 0, 1/2), (1/3, 1/4, 1/4, 1/9, 1/18)], 
and  the generators of $G^T$ are  [(0, 0, 0, 1/9, 8/9), (1/3, 1/4, 1/4, 1/9, 1/18)].

There are $22$ $G$-invariant deformations, 
(1, 1, 1, 1, 1) 
(0, 0, 0, 1, 16) (0, 0, 0, 2, 14) (0, 0, 0, 3, 12) (0, 0, 0, 4, 10) 
(0, 0, 0, 5, 8) (0, 0, 0, 6, 6) (0, 0, 0, 7, 4) (0, 1, 1, 0, 9) 
(0, 1, 1, 1, 7) (0, 1, 1, 2, 5) (0, 1, 1, 3, 3) (0, 1, 1, 4, 1) 
(0, 2, 2, 0, 0) (1, 0, 0, 0, 12) (1, 0, 0, 1, 10) (1, 0, 0, 2, 8) 
(1, 0, 0, 3, 6) (1, 0, 0, 4, 4) (1, 0, 0, 5, 2) (1, 0, 0, 6, 0) 
(1, 1, 1, 0, 3). 

Then are $14$ $G^T$-invariant deformations, 
(1, 1, 1, 1, 1) 
(0, 0, 0, 3, 12) (0, 0, 0, 6, 6) (0, 0, 2, 0, 9) (0, 0, 2, 3, 3) 
(0, 1, 1, 0, 9) (0, 1, 1, 3, 3) (0, 2, 0, 0, 9) (0, 2, 0, 3, 3) 
(0, 2, 2, 0, 0) (1, 0, 0, 1, 10) (1, 0, 0, 4, 4) (1, 0, 2, 1, 1) 
(1, 2, 0, 1, 1). 


There are five $(G,G^T)$-Roans pairs, 
[(0, 1/4, 1/4, 0, 1/2), (1/3, 0, 0, 2/3, 0)] 
[(0, 1/4, 1/4, 0, 1/2), (2/3, 0, 0, 1/3, 0)] [(1/3, 0, 0, 1/9, 5/9), (0, 1/2, 1/2, 0, 0)] [(1/3, 0, 0, 4/9, 2/9), (0, 1/2, 1/2, 0, 0)] [(2/3, 0, 0, 2/9, 1/9), (0, 1/2, 1/2, 0, 0)] 

There are $21$ $(G^T,G)$-Roans pairs, 
[(1/3, 0, 0, 1/9, 5/9), (0, 1/4, 3/4, 0, 0)] 
[(1/3, 0, 0, 1/9, 5/9), (0, 1/2, 1/2, 0, 0)] [(1/3, 0, 0, 1/9, 5/9), (0, 3/4, 1/4, 0, 0)] [(1/3, 0, 0, 2/9, 4/9), (0, 1/4, 3/4, 0, 0)] [(1/3, 0, 0, 2/9, 4/9), (0, 1/2, 1/2, 0, 0)] 
[(1/3, 0, 0, 2/9, 4/9), (0, 3/4, 1/4, 0, 0)] [(1/3, 0, 0, 1/3, 1/3), (0, 1/4, 3/4, 0, 0)] [(1/3, 0, 0, 1/3, 1/3), (0, 1/2, 1/2, 0, 0)] [(1/3, 0, 0, 1/3, 1/3), (0, 3/4, 1/4, 0, 0)] 
[(1/3, 0, 0, 4/9, 2/9), (0, 1/4, 3/4, 0, 0)] [(1/3, 0, 0, 4/9, 2/9), (0, 1/2, 1/2, 0, 0)] [(1/3, 0, 0, 4/9, 2/9), (0, 3/4, 1/4, 0, 0)] [(1/3, 0, 0, 5/9, 1/9), (0, 1/4, 3/4, 0, 0)] 
[(1/3, 0, 0, 5/9, 1/9), (0, 1/2, 1/2, 0, 0)] [(1/3, 0, 0, 5/9, 1/9), (0, 3/4, 1/4, 0, 0)] [(2/3, 0, 0, 1/9, 2/9), (0, 1/4, 3/4, 0, 0)] [(2/3, 0, 0, 1/9, 2/9), (0, 1/2, 1/2, 0, 0)] 
[(2/3, 0, 0, 1/9, 2/9), (0, 3/4, 1/4, 0, 0)] [(2/3, 0, 0, 2/9, 1/9), (0, 1/4, 3/4, 0, 0)] [(2/3, 0, 0, 2/9, 1/9), (0, 1/2, 1/2, 0, 0)] [(2/3, 0, 0, 2/9, 1/9), (0, 3/4, 1/4, 0, 0)].

Thus  $h^{1,1}(\widetilde{W_{3,4,4,9,18}/G}=19=5-4+18$, $h^{2,1}(\widetilde{W_{3,4,4,9,18}/G}=43=65-4-18$.

\section{The Roan pairs for orbifolds of Calabi-Yau surfaces Fermat type ($K3$ surfaces of Fermat type) }

Recall that two dimensional Calabi-Yau are $K3$ surfaces. The Hodge diamond of $K3$ surfaces is the following 
\[
\begin{array}{cccccccc}
& & 1& & \cr
&0& & 0 & &\cr
1& &20 & & 1\cr
&0& & 0 & &\cr
& & 1& & \cr
\end{array}
\]  
Here, for any $K3$ surface $S$ of Fermat type,   we obtain decomposition of the cohomology  $H^2(\widetilde{S/G}, \mathbb Z)$ for all orbifolds $S$, where  $\widetilde{S/G}$ denote desingularization of $S/G$, $G\subset G^{\max}$,
 into untwisted and twisted parts.  
 
 Let us note, that  for all orbifolds of $K3$ surfaces Fermat type, the Hodge diamond is   as above.
 
 For this we define the Roan pairs. 
 
 \begin{dfn}
A pair  $(\mathbf z', \mathbf z)\in G^T_{2Z}\times G_{Z}$ is a $(G^T, G)$- {\em Roan's pair} if 

\[
\sum  z'_i\,a_i\, z_i=0. 
\]

A pair   $(\mathbf z, \mathbf z')\in G_{2Z}\times G^T_{2Z}$ is a $(G, G^T)-$ {\em Roan's pair} if 
\[
\sum z_i\,a_i \,z'_i=0.
\]

     For any Fermat CY $K3$ surface  and a group $G_0\subset G\subset G^{\max}$, we set 
\begin{equation}\label{T1}
\mbox{untiwisted  part of } H^{1,1}(\widetilde{S/G})=\{ G-\mbox{ invariant deformations}\} ;
\end{equation}

\begin{equation}\label{T2}
\mbox{twisted  part of }H^{1,1}(\widetilde{S/G})=\{ G^T-\mbox{invariant deformations}\} \cup 
\{ (G^T, G) \mbox{ Roan's pairs}\}. 
\end{equation}
\end{dfn}

Due to the above definition, we get that under the BHK mirror symmetry,  the untwisted part is invariant and, for the twisted,   $(G^T, G)$ Roan's pairs swap with $(G, G^T)$ Roan's pairs.

\begin{thm}\label{K3Roan}
For any $K3$ surface $S$ of Fermat type and any orbifold $S/G$, for the Euler numbers we have $\chi(\widetilde{S/G})=24$ according to the Vafa formula, and 
there holds 
\[
20=h^{1,1}=|\mbox{untiwisted  part of } H^{1,1}(\widetilde{S/G})|+ |\mbox{twisted  part of }H^{1,1}(\widetilde{S/G})|,
\]
where untwisted and twisted parts are (\ref{T1}) and (\ref{T2}), respectively.
\end{thm}
{\em Proof}. The proof is going by case-by-case checking, and  documented in \cite{DBK3}.
\hfill $\Box$
 
\begin{ex}
Consider $K3$ surface 
\[
w^4+x^4+y^4+z^4=0
\]
in $\mathbb P(1,1,1,1)$.
\end{ex}
Then \begin{enumerate}
\item
For $G=G^{\max}$, $G=Z_4\times Z_4\times Z_4$,  we have the following.

The generators $G$ are [(0, 0, 1/4, 3/4), (0, 1/4, 0, 3/4), (1/4, 1/4, 1/4, 1/4)], and $G_0=Z_4$ has the generator [(1/4, 1/4, 1/4, 1/4)].

There is only one $G$-invariant deformation  (1, 1, 1, 1)  and $19$ $G_0$-invariant deformations, 
(1, 1, 1, 1) 
(0, 0, 2, 2) (0, 1, 1, 2) (0, 1, 2, 1) (0, 2, 0, 2) 
(0, 2, 1, 1) (0, 2, 2, 0) (1, 0, 1, 2) (1, 0, 2, 1) 
(1, 1, 0, 2) (1, 1, 2, 0) (1, 2, 0, 1) (1, 2, 1, 0) 
(2, 0, 0, 2) (2, 0, 1, 1) (2, 0, 2, 0) (2, 1, 0, 1) 
(2, 1, 1, 0) (2, 2, 0, 0).

There are no Roan's pairs. 

\item
For $G=Z_4\times Z_4$ with generators  [(0, 1/2, 1/4, 1/4), (1/4, 1/4, 1/4, 1/4)], $G^T=Z_4\times Z_4$ has generators 
[(0, 0, 1/4, 3/4), (1/4, 1/4, 1/4, 1/4)].

Then there are five $G$-invariant deformations,
(1, 1, 1, 1) 
(0, 0, 2, 2) (1, 1, 0, 2) (1, 1, 2, 0) (2, 2, 0, 0);

there are five $G^T$-invariant deformations
(1, 1, 1, 1) 
(0, 0, 2, 2) (0, 2, 1, 1) (2, 0, 1, 1) (2, 2, 0, 0).

There are ten  $(G,G^T)$-Roan's pairs 
[(0, 0, 1/2, 1/2), (1/2, 1/2, 0, 0)] 
[(1/4, 3/4, 0, 0), (0, 0, 1/4, 3/4)] [(1/4, 3/4, 0, 0), (0, 0, 1/2, 1/2)] [(1/4, 3/4, 0, 0), (0, 0, 3/4, 1/4)] [(1/2, 1/2, 0, 0), (0, 0, 1/4, 3/4)] 
[(1/2, 1/2, 0, 0), (0, 0, 1/2, 1/2)] [(1/2, 1/2, 0, 0), (0, 0, 3/4, 1/4)] [(3/4, 1/4, 0, 0), (0, 0, 1/4, 3/4)] [(3/4, 1/4, 0, 0), (0, 0, 1/2, 1/2)] 
[(3/4, 1/4, 0, 0), (0, 0, 3/4, 1/4)];

and  ten  $(G^T,G)$-Roan's pairs 
[(0, 0, 1/4, 3/4), (1/4, 3/4, 0, 0)] 
[(0, 0, 1/4, 3/4), (1/2, 1/2, 0, 0)] [(0, 0, 1/4, 3/4), (3/4, 1/4, 0, 0)] [(0, 0, 1/2, 1/2), (1/4, 3/4, 0, 0)] [(0, 0, 1/2, 1/2), (1/2, 1/2, 0, 0)] 
[(0, 0, 1/2, 1/2), (3/4, 1/4, 0, 0)] [(0, 0, 3/4, 1/4), (1/4, 3/4, 0, 0)] [(0, 0, 3/4, 1/4), (1/2, 1/2, 0, 0)] [(0, 0, 3/4, 1/4), (3/4, 1/4, 0, 0)] 
[(1/2, 1/2, 0, 0), (0, 0, 1/2, 1/2)] 

\item
For $G=Z_4\times Z_4\times Z_2$ with generators [(0, 0, 1/4, 3/4), (0, 1/2, 0, 1/2), (1/4, 1/4, 1/4, 1/4)], the mirror group $G^T=Z_4\times Z_2$
has generators  [(0, 0, 1/2, 1/2), (1/4, 1/4, 1/4, 1/4)].

Then there are tree $G$-invariant deformations, 
(1, 1, 1, 1) 
(0, 0, 2, 2) (2, 2, 0, 0) ;

there are $11$ $G^T$-invariant deformations
(1, 1, 1, 1) 
(0, 0, 2, 2) (0, 2, 0, 2) (0, 2, 1, 1) (0, 2, 2, 0) 
(1, 1, 0, 2) (1, 1, 2, 0) (2, 0, 0, 2) (2, 0, 1, 1)
(2, 0, 2, 0) (2, 2, 0, 0);

and there are six  $(G,G^T)$-Roan's pairs, 

[(0, 0, 1/4, 3/4), (1/2, 1/2, 0, 0)] 
[(0, 0, 1/2, 1/2), (1/2, 1/2, 0, 0)] [(0, 0, 3/4, 1/4), (1/2, 1/2, 0, 0)] [(1/4, 3/4, 0, 0), (0, 0, 1/2, 1/2)] [(1/2, 1/2, 0, 0), (0, 0, 1/2, 1/2)] 
[(3/4, 1/4, 0, 0), (0, 0, 1/2, 1/2)]; 

and there are six  $(G^T,G)$-Roan's pairs,
[(0, 0, 1/2, 1/2), (1/4, 3/4, 0, 0)] 
[(0, 0, 1/2, 1/2), (1/2, 1/2, 0, 0)] [(0, 0, 1/2, 1/2), (3/4, 1/4, 0, 0)] [(1/2, 1/2, 0, 0), (0, 0, 1/4, 3/4)] [(1/2, 1/2, 0, 0), (0, 0, 1/2, 1/2)] 
[(1/2, 1/2, 0, 0), (0, 0, 3/4, 1/4)] 
\item 
For $G=Z_4$ with generator [(1/4, 1/4, 1/4, 1/4)], and the mirror $G^T=Z_4\times Z_4\times Z_4$ with generators. [(0, 0, 1/4, 3/4), (0, 1/4, 0, 3/4), (1/4, 1/4, 1/4, 1/4)]

We have $19$ $G$-invariant deformations, 
(1, 1, 1, 1) 
(0, 0, 2, 2) (0, 1, 1, 2) (0, 1, 2, 1) (0, 2, 0, 2) 
(0, 2, 1, 1) (0, 2, 2, 0) (1, 0, 1, 2) (1, 0, 2, 1) 
(1, 1, 0, 2) (1, 1, 2, 0) (1, 2, 0, 1) (1, 2, 1, 0) 
(2, 0, 0, 2) (2, 0, 1, 1) (2, 0, 2, 0) (2, 1, 0, 1) 
(2, 1, 1, 0) (2, 2, 0, 0) ;

one $G^T$-invariant deformation 
(1, 1, 1, 1),

and no Roan's pairs.

This is the mirror to first orbifold.

\item For 
$G=Z_4\times Z_2$ with generators  [(0, 0, 1/2, 1/2), (1/4, 1/4, 1/4, 1/4)], the mirror $G^T$ has generators [(0, 0, 1/4, 3/4), (0, 1/2, 0, 1/2), (1/4, 1/4, 1/4, 1/4)], 

there are $11$ $G$-invariant deformations 
(1, 1, 1, 1) 
(0, 0, 2, 2) (0, 2, 0, 2) (0, 2, 1, 1) (0, 2, 2, 0) 
(1, 1, 0, 2) (1, 1, 2, 0) (2, 0, 0, 2) (2, 0, 1, 1) 
(2, 0, 2, 0) (2, 2, 0, 0);

there are three $G^T$-invariant deformations 
(0, 0, 2, 2) (2, 2, 0, 0);

and there are six  $(G,G^T)$-Roan's pairs
[(0, 0, 1/2, 1/2), (1/4, 3/4, 0, 0)] 
[(0, 0, 1/2, 1/2), (1/2, 1/2, 0, 0)] [(0, 0, 1/2, 1/2), (3/4, 1/4, 0, 0)] [(1/2, 1/2, 0, 0), (0, 0, 1/4, 3/4)] [(1/2, 1/2, 0, 0), (0, 0, 1/2, 1/2)] 
[(1/2, 1/2, 0, 0), (0, 0, 3/4, 1/4)];

and there are six  $(G^T,G)$-Roan's pairs
[(0, 0, 1/4, 3/4), (1/2, 1/2, 0, 0)] 
[(0, 0, 1/2, 1/2), (1/2, 1/2, 0, 0)] [(0, 0, 3/4, 1/4), (1/2, 1/2, 0, 0)] [(1/4, 3/4, 0, 0), (0, 0, 1/2, 1/2)] [(1/2, 1/2, 0, 0), (0, 0, 1/2, 1/2)] 
[(3/4, 1/4, 0, 0), (0, 0, 1/2, 1/2)].

This is mirror to the third orbifold.
\item
For $G=Z_4\times Z_4$ with the generators  [(0, 0, 1/2, 1/2), (0, 1/2, 0, 1/2), (1/4, 1/4, 1/4, 1/4)], the mirror $G^T=Z_4\times Z_4$ with generators  [(0, 0, 1/2, 1/2), (0, 1/2, 0, 1/2), (1/4, 1/4, 1/4, 1/4)],
we have 
seven $G$-invariant deformations, 
(1, 1, 1, 1) 
(0, 0, 2, 2) (0, 2, 0, 2) (0, 2, 2, 0) (2, 0, 0, 2) 
(2, 0, 2, 0) (2, 2, 0, 0) ;

seven $G^T$-invariant deformations, 
(1, 1, 1, 1) 
(0, 0, 2, 2) (0, 2, 0, 2) (0, 2, 2, 0) (2, 0, 0, 2) 
(2, 0, 2, 0) (2, 2, 0, 0) ;

and six  $(G,G^T)$-Roan's pairs 
[(0, 0, 1/2, 1/2), (1/2, 1/2, 0, 0)] 
[(0, 1/2, 0, 1/2), (1/2, 0, 1/2, 0)] [(0, 1/2, 1/2, 0), (1/2, 0, 0, 1/2)] [(1/2, 0, 0, 1/2), (0, 1/2, 1/2, 0)] [(1/2, 0, 1/2, 0), (0, 1/2, 0, 1/2)] 
[(1/2, 1/2, 0, 0), (0, 0, 1/2, 1/2)], 

and six  $(G^T,G)$-Roan's pairs 
[(0, 0, 1/2, 1/2), (1/2, 1/2, 0, 0)] 
[(0, 1/2, 0, 1/2), (1/2, 0, 1/2, 0)] [(0, 1/2, 1/2, 0), (1/2, 0, 0, 1/2)] [(1/2, 0, 0, 1/2), (0, 1/2, 1/2, 0)] [(1/2, 0, 1/2, 0), (0, 1/2, 0, 1/2)] 
[(1/2, 1/2, 0, 0), (0, 0, 1/2, 1/2)] 

This orbifold is self dual with respect to mirror symmetry.
\end{enumerate}
\begin{ex}
    Consider $K3$ surface 
    \[
    w^2+x^3+y^{10}+z^{15}=0
    \]
    in the weighted projective space $\mathbb P(15, 10,3,2)$.
    Then $G^{\max}=G_0=Z_{30}$.
\end{ex}
There are eight $G^{\max}$-invariant (0, 0, 2, 12) (0, 0, 4, 9) (0, 0, 6, 6) (0, 0, 8, 3) 
(0, 1, 0, 10) (0, 1, 2, 7) (0, 1, 4, 4) (0, 1, 6, 1);

the same for $G_0$; 

there are four  $(G,G^T)$- Roan's pairs
[(0, 1/3, 0, 2/3), (1/2, 0, 1/2, 0)] 
[(0, 2/3, 0, 1/3), (1/2, 0, 1/2, 0)] [(1/2, 0, 1/2, 0), (0, 1/3, 0, 2/3)] [(1/2, 0, 1/2, 0), (0, 2/3, 0, 1/3)] 

and the same set of $(G^,G)$- Roan's pairs. 

Thus, we get the self mirror $K3$ surface as claimed in \cite{Borcea}. 

\section{Conclusion}


Calabi-Yau varieties of the form the product of a Calabi-Yau variety and ellipric curve are of importance for rwelve dimensional $F$-theory \cite{VafaF}. Namely, for dimensions  $D=8$ and $6$, Vafa \cite{VafaF}, argued on dualities 
\begin{equation}\label{FT1}
F(K3)\leftrightarrow  Het (E),\mbox{ and } F(CY_3) \leftrightarrow  Het (K3)
\end{equation}
respectively, where $CY_k$ denotes Calabi-Yau $k$-fold, 
$F(CY_n)$ denotes F-theory in $12$ dimension with compactification on $CY_n$, and $Het(CY_{n-1})$ denotes 10 dimensional heterotic string with compactification $CY_{n-1}$.

In \cite{BS} it was discussed possibilities for duality in dimension $D=4$,
\begin{equation}\label{FT2}
F(CY_4) \leftrightarrow  Het (CY_3).
\end{equation}
For such dualities we propose to $CY_n$ considered as an extension of the above relation Barcea-Voisin construction and BHK mirror symmetry for $CY_3$.

Specifically, for $D=8$ and an elliptic curve, for the duality, we may consider $K3$ surface as the desingularization of product two elliptic curves in $\mathbb P(3,2,1)$ and 
and $\mathbb P(2,1,1)$ (see  of (\ref{E2}) and (\ref{E3}) quotient by the involution of order $2$. The corresponding orbifold of Fermat twofold is obtained from
\[
W=w^4+u^4+y^3+z^6
\]
by quotient by the group $G^{\max}=Z_2\times Z_{12}$. The point is that, for this orbifold, the dimension of untwisted part of $H^{1,1}$ is $4$ and twisted is of dimension $16$, and $16$ is the number of exceptional divisors of the desingularization.

For $D=6$ and ten $K3$ surfaces of Theorem \ref{BV-BHK}, for the F-duality, one may consider the corresponding Fermat orbifolds of this theorem.

For $D=4$, we may expect for $108$ Fermat threefolds of the type $(2, a,b,c,d)$ which will give duality, F-theory for orbifolds of Fermat fourfold of type
$(4,4,a,b,c,d)$ and heterotic string for Fermat threefold of type $(2, a,b,c,d)$. It would be important to have a combinatorial formula for the Euler number of 
such orbifolds of fourfolds.

In \cite{BS} it was proposed to relate a Fermat threefold $(a_1,a_2,a_3,a_4,a_5)$ with the Fermat fourfold $(a_1,a_2,a_3,a_4, 2a_5, 2a_5)$.
Note that if to follow such a way for $K3$ surfaces of Fermat type, no orbifolds of such threefolds will have the Hodge numbers as in Theorem \ref{BV-BHK}. This can be directly checked of \cite{DB}.

Recall, that the Euler number of fourfold CY $X$ is computed by the following formula
\[
\chi(X)=6(8+h^{1,1}(X)+h^{3,1}(X)-h^{2,1}(X)),
\]
where
$h^{3,1}$ counts the complex structures and $h^{1,1}$ counts K\"ahler moduli.

For orbifolds Fermat CY  fourfolds we expect that the Hodge numbers and  $G$, $G^T$-invariant deformations and the Roan pairs are related as following.

For Fermat CY fourfold 
\[W=z_1^{a_1}+\cdots +z_6^{a_6}
\]
and a admissible group $G_0\subset G\subset G^{\max}$, we  define the Roan pairs $R^{4,2}$, $R^{3,3}$ and $R^{2,4}$.

The set $R^{4,2}$ is constituted from pairs $\mathbf z, \mathbf z'$ of elements of level one of $G$ and $G^T$, respectively, such that 
$\mathbf z$ has 4 zero coordinates, and $\mathbf z'$ has 2 zero coordinates and they are orthogonal.

The set $R^{2,4}$ is constituted from pairs $\mathbf z, \mathbf z'$ of elements of level one of $G$ and $G^T$, respectively, such that 
$\mathbf z$ has 2 zero coordinates, and $\mathbf z'$ has 4 zero coordinates and they are orthogonal.

The set $R^{3,3}$ is constituted from pairs $\mathbf z, \mathbf z'$ of elements of level one of $G$ and $G^T$, respectively, such that 
$\mathbf z$ has 3 zero coordinates, and $\mathbf z'$ has 3 zero coordinates and they are orthogonal.

Then 
\[
h^{3,1}=|G-\mbox{invariant deformations}|+|R^{2,4}|;\]
\[
h^{1,1}=|G^T-\mbox{invariant deformations}|+|R^{4,2}|;\]
\[
h^{2,1}=|R^{3,3}|.\]

For threefolds, we found (see \cite{DB}) that neither log-convexity nor log-concavity hold for Euler numbers. 

For example,  for  Fermat of type (2,3,8,32,96), Number 20 in \cite{DB}, we have triples of Euler numbers 
\[ (408, 168, 48) \mbox{ with } 168^2> 408\times 48.
\]
and 
\[ (408, 144, 48) \mbox{ with } 144^2> 408\times 48.
\]

 For  Fermat of type (2,3,8,36,72), Number 22 in \cite{DB}, we have triple of Euler numbers 
\[ (456, 192, 96) \mbox{ with } 192^2< 456\times 96.
\]

 


\section*{Appendix}
Here we give the list of deformations and the Roan pairs for orbifolds of K3 surfaces of the Fermat type
\subsection{(4,4,4,4)}
$A = (4, 4, 4, 4)$ in $ \mathbb P (1, 1, 1, 1)$.
\begin{enumerate}
    \item 
$|G| = 64$, generators  [(0, 0, 1/4, 3/4), (0, 1/4, 0, 3/4), (1/4, 1/4, 1/4, 1/4)],
the generators of the mirror group $G^T=G_0=Z_4$,  [ (1/4, 1/4, 1/4, 1/4)]

There is only one $G$-invariant deformation
(1, 1, 1, 1), and 19 $G^T$-invariant deformation, 
(1, 1, 1, 1) 
(0, 0, 2, 2) (0, 1, 1, 2) (0, 1, 2, 1) (0, 2, 0, 2) 
(0, 2, 1, 1) (0, 2, 2, 0) (1, 0, 1, 2) (1, 0, 2, 1) 
(1, 1, 0, 2) (1, 1, 2, 0) (1, 2, 0, 1) (1, 2, 1, 0) 
(2, 0, 0, 2) (2, 0, 1, 1) (2, 0, 2, 0) (2, 1, 0, 1) 
(2, 1, 1, 0) (2, 2, 0, 0). 

There are no Roan's pairs.

$\chi_{Vafa}(W_0/G) = 24$, $h^{1,1}=1+19=20$. 

\item 
$|G| = 16$ with generators  [(0, 1/2, 1/4, 1/4), (1/4, 1/4, 1/4, 1/4)].
The mirror groups has generators [(0, 0, 1/4, 3/4), (1/4, 1/4, 1/4, 1/4)].

There are 5 $G$-deformationa 
(1, 1, 1, 1) 
(0, 0, 2, 2) (1, 1, 0, 2) (1, 1, 2, 0) (2, 2, 0, 0);

and 5 $G^T$ deformations 
(1, 1, 1, 1) 
(0, 0, 2, 2) (0, 2, 1, 1) (2, 0, 1, 1) (2, 2, 0, 0).

There are 10 $(G,G^T)$-Roan's pairs
[(0, 0, 1/2, 1/2), (1/2, 1/2, 0, 0)] 
[(1/4, 3/4, 0, 0), (0, 0, 1/4, 3/4)] [(1/4, 3/4, 0, 0), (0, 0, 1/2, 1/2)] [(1/4, 3/4, 0, 0), (0, 0, 3/4, 1/4)] [(1/2, 1/2, 0, 0), (0, 0, 1/4, 3/4)] 
[(1/2, 1/2, 0, 0), (0, 0, 1/2, 1/2)] [(1/2, 1/2, 0, 0), (0, 0, 3/4, 1/4)] [(3/4, 1/4, 0, 0), (0, 0, 1/4, 3/4)] [(3/4, 1/4, 0, 0), (0, 0, 1/2, 1/2)] 
[(3/4, 1/4, 0, 0), (0, 0, 3/4, 1/4)].

And 10 $(G^,G)$-Roan's pairs
[(0, 0, 1/4, 3/4), (1/4, 3/4, 0, 0)] 
[(0, 0, 1/4, 3/4), (1/2, 1/2, 0, 0)] [(0, 0, 1/4, 3/4), (3/4, 1/4, 0, 0)] [(0, 0, 1/2, 1/2), (1/4, 3/4, 0, 0)] [(0, 0, 1/2, 1/2), (1/2, 1/2, 0, 0)] 
[(0, 0, 1/2, 1/2), (3/4, 1/4, 0, 0)] [(0, 0, 3/4, 1/4), (1/4, 3/4, 0, 0)] [(0, 0, 3/4, 1/4), (1/2, 1/2, 0, 0)] [(0, 0, 3/4, 1/4), (3/4, 1/4, 0, 0)] 
[(1/2, 1/2, 0, 0), (0, 0, 1/2, 1/2)].

$\chi_{Vafa}(W_0/G) = 24$, $h^{1,1}=5+5+10=20$.


\item 
$|G| = 32$ with generators  [(0, 0, 1/4, 3/4), (0, 1/2, 0, 1/2), (1/4, 1/4, 1/4, 1/4)], and $G^T$ with generators [(0, 0, 1/2, 1/2), (1/4, 1/4, 1/4, 1/4)]

Then, there are 3 $G$-invariant deformations 
(1, 1, 1, 1) 
(0, 0, 2, 2) (2, 2, 0, 0);

and there are 11 $G^T$-invariant deformations 
(1, 1, 1, 1) 
(0, 0, 2, 2) (0, 2, 0, 2) (0, 2, 1, 1) (0, 2, 2, 0) 
(1, 1, 0, 2) (1, 1, 2, 0) (2, 0, 0, 2) (2, 0, 1, 1) 
(2, 0, 2, 0) (2, 2, 0, 0) 

There are $6$ $(G,G^T)$-pairs, 
[(0, 0, 1/4, 3/4), (1/2, 1/2, 0, 0)] 
[(0, 0, 1/2, 1/2), (1/2, 1/2, 0, 0)] [(0, 0, 3/4, 1/4), (1/2, 1/2, 0, 0)] [(1/4, 3/4, 0, 0), (0, 0, 1/2, 1/2)] [(1/2, 1/2, 0, 0), (0, 0, 1/2, 1/2)] 
[(3/4, 1/4, 0, 0), (0, 0, 1/2, 1/2)];

and $6$ $(G^,G)$-pairs,
[(0, 0, 1/2, 1/2), (1/4, 3/4, 0, 0)] 
[(0, 0, 1/2, 1/2), (1/2, 1/2, 0, 0)] [(0, 0, 1/2, 1/2), (3/4, 1/4, 0, 0)] [(1/2, 1/2, 0, 0), (0, 0, 1/4, 3/4)] [(1/2, 1/2, 0, 0), (0, 0, 1/2, 1/2)] 
[(1/2, 1/2, 0, 0), (0, 0, 3/4, 1/4)] 

$\chi_{Vafa}(W_0/G) = 24$, $h^{1,1}=3+11+6=20$.

\end{enumerate}

\subsection{(3,4,4,6)}
$A = (3, 4, 4, 6)$ in $\mathbb P (4, 3, 3, 2)$.
\begin{enumerate}
    \item 
$|G| = 24$ with generators [(0, 0, 1/2, 1/2), (1/3, 1/4, 1/4, 1/6)] and $G^T=G_0=Z_{12}$ with generator (1/3, 1/4, 1/4, 1/6).

There are $4$ $G$-deformations, 
(1, 1, 1, 1) 
(0, 1, 1, 3) (0, 2, 2, 0) (1, 0, 0, 4);

and there are $8$ $G^T$-deformations 
(1, 1, 1, 1) 
(0, 0, 2, 3) (0, 1, 1, 3) (0, 2, 0, 3) (0, 2, 2, 0) 
(1, 0, 0, 4) (1, 0, 2, 1) (1, 2, 0, 1). 

There are $8$ $(G,G^T)$-pairs, 
[(0, 1/4, 3/4, 0), (1/3, 0, 0, 2/3)] 
[(0, 1/4, 3/4, 0), (2/3, 0, 0, 1/3)] [(0, 1/2, 1/2, 0), (1/3, 0, 0, 2/3)] [(0, 1/2, 1/2, 0), (2/3, 0, 0, 1/3)] [(0, 3/4, 1/4, 0), (1/3, 0, 0, 2/3)] 
[(0, 3/4, 1/4, 0), (2/3, 0, 0, 1/3)] [(1/3, 0, 0, 2/3), (0, 1/2, 1/2, 0)] [(2/3, 0, 0, 1/3), (0, 1/2, 1/2, 0)];

and there are $8$ $(G^,G)$-pairs, 
[(0, 1/2, 1/2, 0), (1/3, 0, 0, 2/3)] 
[(0, 1/2, 1/2, 0), (2/3, 0, 0, 1/3)] [(1/3, 0, 0, 2/3), (0, 1/4, 3/4, 0)] [(1/3, 0, 0, 2/3), (0, 1/2, 1/2, 0)] [(1/3, 0, 0, 2/3), (0, 3/4, 1/4, 0)] 
[(2/3, 0, 0, 1/3), (0, 1/4, 3/4, 0)] [(2/3, 0, 0, 1/3), (0, 1/2, 1/2, 0)] [(2/3, 0, 0, 1/3), (0, 3/4, 1/4, 0)] 

$\chi_{Vafa}(W_0/G) = 24$, $h^{1,1}=4+8+8=20$.

\end{enumerate}
\subsection{(3,3,6,6)}
$A = (3, 3, 6, 6)$ in $\mathbb P (2, 2, 1, 1)$.

\begin{enumerate}
    \item 
$|G^{\max}| = 54$ with generators  [(0, 0, 1/6, 5/6), (0, 1/3, 0, 2/3), (1/3, 1/3, 1/6, 1/6)],
$G^T=G_0=Z_6$ with generator  (1/3, 1/3, 1/6, 1/6).

There are 2 $G$-deformations 
(1, 1, 1, 1), 
(0, 0, 3, 3);

and $16$ $G^T$-deformations
(1, 1, 1, 1) 
(0, 0, 2, 4) (0, 0, 3, 3) (0, 0, 4, 2) (0, 1, 0, 4) 
(0, 1, 1, 3) (0, 1, 2, 2) (0, 1, 3, 1) (0, 1, 4, 0) 
(1, 0, 0, 4) (1, 0, 1, 3) (1, 0, 2, 2) (1, 0, 3, 1) 
(1, 0, 4, 0) (1, 1, 0, 2) (1, 1, 2, 0);

there are two $(G,G^T)$-pairs, 
[(1/3, 2/3, 0, 0), (0, 0, 1/2, 1/2)] 
[(2/3, 1/3, 0, 0), (0, 0, 1/2, 1/2)];

and there are two $(G^T,G)$-pairs, 
[(0, 0, 1/2, 1/2), (1/3, 2/3, 0, 0)] 
[(0, 0, 1/2, 1/2), (2/3, 1/3, 0, 0)].

$\chi_{Vafa}(W_0/G) = 24$, $h^{1,1}=2+16+2=20$.


\item 
$|G| = 18$ with generators  [(0, 1/3, 1/3, 1/3), (1/3, 1/3, 1/6, 1/6)], $G^T$ has generators [(0, 0, 1/6, 5/6), (1/3, 1/3, 1/6, 1/6)]

There are $6$ $G$-deformations 
(1, 1, 1, 1) 
(0, 0, 2, 4) (0, 0, 3, 3) (0, 0, 4, 2) (1, 1, 0, 2) 
(1, 1, 2, 0);

and $4$ $G^T$-deformations 
(1, 1, 1, 1) 
(0, 0, 3, 3) (0, 1, 2, 2) (1, 0, 2, 2) 

There are 10 $(G,G^T)$-pairs, 

[(1/3, 2/3, 0, 0), (0, 0, 1/6, 5/6)] 
[(1/3, 2/3, 0, 0), (0, 0, 1/3, 2/3)] [(1/3, 2/3, 0, 0), (0, 0, 1/2, 1/2)] [(1/3, 2/3, 0, 0), (0, 0, 2/3, 1/3)] [(1/3, 2/3, 0, 0), (0, 0, 5/6, 1/6)] 
[(2/3, 1/3, 0, 0), (0, 0, 1/6, 5/6)] [(2/3, 1/3, 0, 0), (0, 0, 1/3, 2/3)] [(2/3, 1/3, 0, 0), (0, 0, 1/2, 1/2)] [(2/3, 1/3, 0, 0), (0, 0, 2/3, 1/3)] 
[(2/3, 1/3, 0, 0), (0, 0, 5/6, 1/6)];

and there are ten $(G^T,G)$-pairs, 
[(0, 0, 1/6, 5/6), (1/3, 2/3, 0, 0)] 
[(0, 0, 1/6, 5/6), (2/3, 1/3, 0, 0)] [(0, 0, 1/3, 2/3), (1/3, 2/3, 0, 0)] [(0, 0, 1/3, 2/3), (2/3, 1/3, 0, 0)] [(0, 0, 1/2, 1/2), (1/3, 2/3, 0, 0)] 
[(0, 0, 1/2, 1/2), (2/3, 1/3, 0, 0)] [(0, 0, 2/3, 1/3), (1/3, 2/3, 0, 0)] [(0, 0, 2/3, 1/3), (2/3, 1/3, 0, 0)] [(0, 0, 5/6, 1/6), (1/3, 2/3, 0, 0)] 
[(0, 0, 5/6, 1/6), (2/3, 1/3, 0, 0)].

$\chi_{Vafa}(W_0/G) = 24$, $h^{1,1}=6+4+10=20$.

\item 
$|G| = 18$ with generators [(0, 1/3, 0, 2/3), (1/3, 1/3, 1/6, 1/6)] and $G^T$ has generators
 [(0, 1/3, 0, 2/3), (1/3, 1/3, 1/6, 1/6)].
 
There are $6$ $G$-deformations 
(1, 1, 1, 1) 
(0, 0, 3, 3) (0, 1, 0, 4) (0, 1, 3, 1) (1, 0, 1, 3) 
(1, 0, 4, 0).

There are $6$ $G^T$-deformations 
(1, 1, 1, 1) 
(0, 0, 3, 3) (0, 1, 0, 4) (0, 1, 3, 1) (1, 0, 1, 3) 
(1, 0, 4, 0). 

There are eight  $(G,G^T)$-pairs,
[(0, 1/3, 0, 2/3), (1/3, 0, 2/3, 0)] 
[(0, 1/3, 0, 2/3), (2/3, 0, 1/3, 0)] [(0, 2/3, 0, 1/3), (1/3, 0, 2/3, 0)] [(0, 2/3, 0, 1/3), (2/3, 0, 1/3, 0)] [(1/3, 0, 2/3, 0), (0, 1/3, 0, 2/3)] 
[(1/3, 0, 2/3, 0), (0, 2/3, 0, 1/3)] [(2/3, 0, 1/3, 0), (0, 1/3, 0, 2/3)] [(2/3, 0, 1/3, 0), (0, 2/3, 0, 1/3)];

and there are 8 $(G^T,G)$-pairs,

[(0, 1/3, 0, 2/3), (1/3, 0, 2/3, 0)] 
[(0, 1/3, 0, 2/3), (2/3, 0, 1/3, 0)] [(0, 2/3, 0, 1/3), (1/3, 0, 2/3, 0)] [(0, 2/3, 0, 1/3), (2/3, 0, 1/3, 0)] [(1/3, 0, 2/3, 0), (0, 1/3, 0, 2/3)] 
[(1/3, 0, 2/3, 0), (0, 2/3, 0, 1/3)] [(2/3, 0, 1/3, 0), (0, 1/3, 0, 2/3)] [(2/3, 0, 1/3, 0), (0, 2/3, 0, 1/3)] 

$\chi_{Vafa}(W_0/G) = 24$, $h^{1,1}=6+6+8=20$.

\end{enumerate}

\subsection{(3, 3, 4, 12)}
$A = (3, 3, 4, 12)$ in  $\mathbb P  (4, 4, 3, 1)$.

\begin{enumerate}
\item 
$G^{\max}$, $|G| = 36$, generators  [(0, 1/3, 0, 2/3), (1/3, 1/3, 1/4, 1/12)];
$G^T=G_0$ with the generators  (1/3, 1/3, 1/4, 1/12).

There are $4$ $G$-deformations 
(1, 1, 1, 1) 
(0, 0, 1, 9) (0, 0, 2, 6) (1, 1, 0, 4);

there are $10$ $G^T$-deformations,
(1, 1, 1, 1) 
(0, 0, 1, 9) (0, 0, 2, 6) (0, 1, 0, 8) (0, 1, 1, 5) 
(0, 1, 2, 2) (1, 0, 0, 8) (1, 0, 1, 5) (1, 0, 2, 2) 
(1, 1, 0, 4). 

There are six  $(G,G^T)$-pairs,
[(1/3, 2/3, 0, 0), (0, 0, 1/4, 3/4)] 
[(1/3, 2/3, 0, 0), (0, 0, 1/2, 1/2)] [(1/3, 2/3, 0, 0), (0, 0, 3/4, 1/4)] [(2/3, 1/3, 0, 0), (0, 0, 1/4, 3/4)] [(2/3, 1/3, 0, 0), (0, 0, 1/2, 1/2)] 
[(2/3, 1/3, 0, 0), (0, 0, 3/4, 1/4)];

and there are $6$  $(G^T,G)$-pairs,
[(0, 0, 1/4, 3/4), (1/3, 2/3, 0, 0)] 
[(0, 0, 1/4, 3/4), (2/3, 1/3, 0, 0)] [(0, 0, 1/2, 1/2), (1/3, 2/3, 0, 0)] [(0, 0, 1/2, 1/2), (2/3, 1/3, 0, 0)] [(0, 0, 3/4, 1/4), (1/3, 2/3, 0, 0)] 
[(0, 0, 3/4, 1/4), (2/3, 1/3, 0, 0)].

$\chi_{Vafa}(W_0/G) = 24$, $h^{1,1}=4+10+6=20$.

\end{enumerate}

\subsection{(2, 6, 6, 6)}

$A = (2, 6, 6, 6)$ in $\mathbb P  (3, 1, 1, 1)$.

\begin{enumerate}
    \item 
    $G=G^{\max}$, $|G| = 72$, generators  [(0, 0, 1/6, 5/6), (0, 1/6, 0, 5/6), (1/2, 1/6, 1/6, 1/6)];

$G^T=G_0$, $G_0=6$, generator  (1/2, 1/6, 1/6, 1/6).

There are $1$ $G$-deformations,
(0, 2, 2, 2). 

There are $19$ $G^T$-deformations, 
(0, 0, 2, 4) (0, 0, 3, 3) (0, 0, 4, 2) (0, 1, 1, 4) 
(0, 1, 2, 3) (0, 1, 3, 2) (0, 1, 4, 1) (0, 2, 0, 4) 
(0, 2, 1, 3) (0, 2, 2, 2) (0, 2, 3, 1) (0, 2, 4, 0) 
(0, 3, 0, 3) (0, 3, 1, 2) (0, 3, 2, 1) (0, 3, 3, 0) 
(0, 4, 0, 2) (0, 4, 1, 1) (0, 4, 2, 0).

There are no Roan's pairs. 

$\chi_{Vafa}(W_0/G) = 24$, $h^{1,1}=1+19=20$.

\item 
$|G| = 12$, its generators are [(0, 0, 1/2, 1/2), (1/2, 1/6, 1/6, 1/6)];
$G^T$, $|G^T|=36$,  has generators  [(0, 0, 1/6, 5/6), (1/2, 1/6, 1/6, 1/6)].

There are $11$ $G$-deformations,
(0, 0, 2, 4) (0, 0, 3, 3) (0, 0, 4, 2) (0, 2, 0, 4) 
(0, 2, 1, 3) (0, 2, 2, 2) (0, 2, 3, 1) (0, 2, 4, 0) 
(0, 4, 0, 2) (0, 4, 1, 1) (0, 4, 2, 0).

There are $3$ $G^T$-deformations,
(0, 0, 3, 3) (0, 2, 2, 2) (0, 4, 1, 1). 

There are six  $(G,G^T)$-pairs,
[(0, 0, 1/2, 1/2), (1/2, 1/2, 0, 0)] 
[(1/2, 1/2, 0, 0), (0, 0, 1/6, 5/6)] [(1/2, 1/2, 0, 0), (0, 0, 1/3, 2/3)] [(1/2, 1/2, 0, 0), (0, 0, 1/2, 1/2)] [(1/2, 1/2, 0, 0), (0, 0, 2/3, 1/3)] 
[(1/2, 1/2, 0, 0), (0, 0, 5/6, 1/6)];

and there are six  $(G^T,G)$-pairs,
[(0, 0, 1/6, 5/6), (1/2, 1/2, 0, 0)] 
[(0, 0, 1/3, 2/3), (1/2, 1/2, 0, 0)] [(0, 0, 1/2, 1/2), (1/2, 1/2, 0, 0)] [(0, 0, 2/3, 1/3), (1/2, 1/2, 0, 0)] [(0, 0, 5/6, 1/6), (1/2, 1/2, 0, 0)] 
[(1/2, 1/2, 0, 0), (0, 0, 1/2, 1/2)] 

$\chi_{Vafa}(W_0/G) = 24$, $h^{1,1}=11+3+6=20$.
\item 
$|G| = 24$, generators  [(0, 0, 1/2, 1/2), (0, 1/6, 1/6, 2/3), (1/2, 1/6, 1/6, 1/6)],
$|G^T| = 18$, generators [(0, 0, 1/3, 2/3), (1/2, 1/6, 1/6, 1/6)].

There are $7$ $G$-deformations,
(0, 0, 2, 4) (0, 0, 4, 2) (0, 2, 0, 4) (0, 2, 2, 2) 
(0, 2, 4, 0) (0, 4, 0, 2) (0, 4, 2, 0).

There are $7$ $G^T$-deformations,
(0, 0, 3, 3) (0, 1, 1, 4) (0, 1, 4, 1) (0, 2, 2, 2) 
(0, 3, 0, 3) (0, 3, 3, 0) (0, 4, 1, 1).

There are six  $(G,G^T)$-pairs,
[(1/2, 0, 0, 1/2), (0, 1/3, 2/3, 0)] 
[(1/2, 0, 0, 1/2), (0, 2/3, 1/3, 0)] [(1/2, 0, 1/2, 0), (0, 1/3, 0, 2/3)] [(1/2, 0, 1/2, 0), (0, 2/3, 0, 1/3)] [(1/2, 1/2, 0, 0), (0, 0, 1/3, 2/3)] 
[(1/2, 1/2, 0, 0), (0, 0, 2/3, 1/3)].

There are six  $(G^T,G)$-pairs,
[(0, 0, 1/3, 2/3), (1/2, 1/2, 0, 0)] 
[(0, 0, 2/3, 1/3), (1/2, 1/2, 0, 0)] [(0, 1/3, 0, 2/3), (1/2, 0, 1/2, 0)] [(0, 1/3, 2/3, 0), (1/2, 0, 0, 1/2)] [(0, 2/3, 0, 1/3), (1/2, 0, 1/2, 0)] 
[(0, 2/3, 1/3, 0), (1/2, 0, 0, 1/2)] 

$\chi_{Vafa}(W_0/G) = 24$, $h^{1,1}=7+7+6=20$.

\end{enumerate}
\subsection{(2,5,5,10)}
$A = (2, 5, 5, 10)$ in $\mathbb P (5, 2, 2, 1)$. 
    $G=G^{\max}$, 
    $|G| = 50$, generators,  [(0, 0, 1/5, 4/5), (1/2, 1/5, 1/5, 1/10)];

    $G^T=G_0=Z_{10}$, generators  (1/2, 1/5, 1/5, 1/10).

 There are $2$ $G^T$-deformations,  
(0, 1, 1, 6) (0, 2, 2, 2).

There are $14$ $G^T$-deformations,
(0, 0, 1, 8) (0, 0, 2, 6) (0, 0, 3, 4) (0, 1, 0, 8) 
(0, 1, 1, 6) (0, 1, 2, 4) (0, 1, 3, 2) (0, 2, 0, 6) 
(0, 2, 1, 4) (0, 2, 2, 2) (0, 2, 3, 0) (0, 3, 0, 4)
(0, 3, 1, 2) (0, 3, 2, 0).

There are four  $(G,G^T)$-pairs,
[(0, 1/5, 4/5, 0), (1/2, 0, 0, 1/2)] 
[(0, 2/5, 3/5, 0), (1/2, 0, 0, 1/2)] [(0, 3/5, 2/5, 0), (1/2, 0, 0, 1/2)] [(0, 4/5, 1/5, 0), (1/2, 0, 0, 1/2)]. 

There are $4$ $(G^T,G)$-pairs,
[(1/2, 0, 0, 1/2), (0, 1/5, 4/5, 0)] 
[(1/2, 0, 0, 1/2), (0, 2/5, 3/5, 0)] [(1/2, 0, 0, 1/2), (0, 3/5, 2/5, 0)] [(1/2, 0, 0, 1/2), (0, 4/5, 1/5, 0)] 

$\chi_{Vafa}(W_0/G) = 24$, $h^{1,1}=2+14+4=20$.

\subsection{(2, 4, 8, 8)}
$A = (2, 4, 8, 8)$, $\mathbb{P} (4, 2, 1, 1)$.

\begin{enumerate}
    \item 
    $G=G^{\max}$, $|G| = 64$, generators [(0, 0, 1/8, 7/8), (0, 1/4, 0, 3/4), (1/2, 1/4, 1/8, 1/8)];

    $G^T=G_0=Z_8$, generators (1/2, 1/4, 1/8, 1/8).

 There are $2$ $G$-deformations, 
(0, 0, 4, 4) (0, 2, 2, 2).

 There are $17$ $G^T$-deformations, 
(0, 0, 2, 6) (0, 0, 3, 5) (0, 0, 4, 4) (0, 0, 5, 3) 
(0, 0, 6, 2) (0, 1, 0, 6) (0, 1, 1, 5) (0, 1, 2, 4) 
(0, 1, 3, 3) (0, 1, 4, 2) (0, 1, 5, 1) (0, 1, 6, 0) 
(0, 2, 0, 4) (0, 2, 1, 3) (0, 2, 2, 2) (0, 2, 3, 1) 
(0, 2, 4, 0) 

There is one  $(G,G^T)$-pair
[(1/2, 1/2, 0, 0), (0, 0, 1/2, 1/2)]  and one $(G^T,G)$-pair
[(0, 0, 1/2, 1/2), (1/2, 1/2, 0, 0)].

$\chi_{Vafa}(W_0/G) = 24$, $h^{1,1}=2+17+1=20$.

\item 
$|G| = 32$, generators  [(0, 0, 1/4, 3/4), (0, 1/4, 1/8, 5/8), (1/2, 1/4, 1/8, 1/8)];

$|G^T|=16$, generators  [(0, 0, 1/4, 3/4), (1/2, 1/4, 1/8, 1/8)].

 There are $6$ $G$-deformations, 
(0, 0, 2, 6) (0, 0, 4, 4) (0, 0, 6, 2) (0, 2, 0, 4) 
(0, 2, 2, 2) (0, 2, 4, 0).

 There are $9$ $G^T$-deformations, 
(0, 0, 2, 6) (0, 0, 4, 4) (0, 0, 6, 2) (0, 1, 1, 5) 
(0, 1, 3, 3) (0, 1, 5, 1) (0, 2, 0, 4) (0, 2, 2, 2) 
(0, 2, 4, 0).

There are five  $(G,G^T)$-pairs,
[(1/2, 0, 0, 1/2), (0, 1/2, 1/2, 0)] 
[(1/2, 0, 1/2, 0), (0, 1/2, 0, 1/2)] [(1/2, 1/2, 0, 0), (0, 0, 1/4, 3/4)] [(1/2, 1/2, 0, 0), (0, 0, 1/2, 1/2)] [(1/2, 1/2, 0, 0), (0, 0, 3/4, 1/4)]. 

There are 5  $(G^T,G)$-pairs, 
[(0, 0, 1/4, 3/4), (1/2, 1/2, 0, 0)] 
[(0, 0, 1/2, 1/2), (1/2, 1/2, 0, 0)] [(0, 0, 3/4, 1/4), (1/2, 1/2, 0, 0)] [(0, 1/2, 0, 1/2), (1/2, 0, 1/2, 0)] [(0, 1/2, 1/2, 0), (1/2, 0, 0, 1/2)] 

$\chi_{Vafa}(W_0/G) = 24$, $h^{1,1}=6+9+5=20$.

\item 
$
|G| = 16$, generators  [(0, 1/4, 3/8, 3/8), (1/2, 1/4, 1/8, 1/8)], $
|G^T| = 32$, generators [(0, 0, 1/8, 7/8), (1/2, 1/4, 1/8, 1/8)].

 There are $10$ $G$-deformations,
(0, 0, 2, 6) (0, 0, 3, 5) (0, 0, 4, 4) (0, 0, 5, 3) 
(0, 0, 6, 2) (0, 2, 0, 4) (0, 2, 1, 3) (0, 2, 2, 2) 
(0, 2, 3, 1) (0, 2, 4, 0).

 There are $3$ $G^T$-deformations,
(0, 0, 4, 4) (0, 1, 3, 3) (0, 2, 2, 2). 

There are seven  $(G,G^T)$-pairs,
[(1/2, 1/2, 0, 0), (0, 0, 1/8, 7/8)] 
[(1/2, 1/2, 0, 0), (0, 0, 1/4, 3/4)] [(1/2, 1/2, 0, 0), (0, 0, 3/8, 5/8)] [(1/2, 1/2, 0, 0), (0, 0, 1/2, 1/2)] [(1/2, 1/2, 0, 0), (0, 0, 5/8, 3/8)] 
[(1/2, 1/2, 0, 0), (0, 0, 3/4, 1/4)] [(1/2, 1/2, 0, 0), (0, 0, 7/8, 1/8)].

There are 7  $(G^T,G)$-pairs, 
[(0, 0, 1/8, 7/8), (1/2, 1/2, 0, 0)] 
[(0, 0, 1/4, 3/4), (1/2, 1/2, 0, 0)] [(0, 0, 3/8, 5/8), (1/2, 1/2, 0, 0)] [(0, 0, 1/2, 1/2), (1/2, 1/2, 0, 0)] [(0, 0, 5/8, 3/8), (1/2, 1/2, 0, 0)] 
[(0, 0, 3/4, 1/4), (1/2, 1/2, 0, 0)] [(0, 0, 7/8, 1/8), (1/2, 1/2, 0, 0)] 

$\chi_{Vafa}(W_0/G) = 24$, $h^{1,1}=10+3+7=20$.
\item 
$
|G| = 16$, generators [(0, 1/4, 1/8, 5/8), (1/2, 1/4, 1/8, 1/8)];

$|G^T| = 32$, generators [(0, 1/4, 0, 3/4), (1/2, 1/4, 1/8, 1/8)].

 There are $10$ $G$-deformations,
(0, 0, 2, 6) (0, 0, 4, 4) (0, 0, 6, 2) (0, 1, 0, 6) 
(0, 1, 2, 4) (0, 1, 4, 2) (0, 1, 6, 0) (0, 2, 0, 4) 
(0, 2, 2, 2) (0, 2, 4, 0). 

 There are $4$ $G^T$-deformations,
(0, 0, 4, 4) (0, 1, 1, 5) (0, 1, 5, 1) (0, 2, 2, 2).

There are six  $(G,G^T)$-pairs,

[(1/2, 0, 0, 1/2), (0, 1/4, 3/4, 0)] 
[(1/2, 0, 0, 1/2), (0, 1/2, 1/2, 0)] [(1/2, 0, 0, 1/2), (0, 3/4, 1/4, 0)] [(1/2, 0, 1/2, 0), (0, 1/4, 0, 3/4)] [(1/2, 0, 1/2, 0), (0, 1/2, 0, 1/2)] 
[(1/2, 0, 1/2, 0), (0, 3/4, 0, 1/4)].

There are 6  $(G^T,G)$-pairs,
[(0, 1/4, 0, 3/4), (1/2, 0, 1/2, 0)] 
[(0, 1/4, 3/4, 0), (1/2, 0, 0, 1/2)] [(0, 1/2, 0, 1/2), (1/2, 0, 1/2, 0)] [(0, 1/2, 1/2, 0), (1/2, 0, 0, 1/2)] [(0, 3/4, 0, 1/4), (1/2, 0, 1/2, 0)] 
[(0, 3/4, 1/4, 0), (1/2, 0, 0, 1/2)] 

$\chi_{Vafa}(W_0/G) = 24$, $h^{1,1}=10+4+6=20$.

\end{enumerate}

\subsection{(2, 4, 6, 12)}

$A = (2, 4, 6, 12)$, $\mathbb P = (6, 3, 2, 1)$.
\begin{enumerate}
    \item $G=G^{\max}$, $|G| = 48$, generators [(0, 0, 1/6, 5/6), (0, 1/4, 0, 3/4), (1/2, 1/4, 1/6, 1/12)];

$G^T=G_0=G_{12}$, generator  (1/2, 1/4, 1/6, 1/12).

 There are $4$ $G$-deformations,
(0, 0, 2, 8) (0, 0, 4, 4) (0, 2, 0, 6) (0, 2, 2, 2).

 There are $13$ $G^T$-deformations,
(0, 0, 1, 10) (0, 0, 2, 8) (0, 0, 3, 6) (0, 0, 4, 4) 
(0, 1, 0, 9) (0, 1, 1, 7) (0, 1, 2, 5) (0, 1, 3, 3) 
(0, 1, 4, 1) (0, 2, 0, 6) (0, 2, 1, 4) (0, 2, 2, 2) 
(0, 2, 3, 0). 

There are three  $(G,G^T)$-pairs,
[(1/2, 0, 1/2, 0), (0, 1/2, 0, 1/2)] 
[(1/2, 1/2, 0, 0), (0, 0, 1/3, 2/3)] [(1/2, 1/2, 0, 0), (0, 0, 2/3, 1/3)];

and there are 3  $(G^T,G)$-pairs,
[(0, 0, 1/3, 2/3), (1/2, 1/2, 0, 0)] 
[(0, 0, 2/3, 1/3), (1/2, 1/2, 0, 0)] [(0, 1/2, 0, 1/2), (1/2, 0, 1/2, 0)] .

$\chi_{Vafa}(W_0/G) = 24$, $h^{1,1}=4 +13+3=20$.

\item 
$|G| = 24$ with generators  [(0, 1/4, 1/6, 7/12), (1/2, 1/4, 1/6, 1/12)], then
$G^T$ has generators [(0, 0, 1/6, 5/6), (1/2, 1/4, 1/6, 1/12)]. 

There are $8$ $G$-deformations,
(0, 0, 1, 10) (0, 0, 2, 8) (0, 0, 3, 6) (0, 0, 4, 4) 
(0, 2, 0, 6) (0, 2, 1, 4) (0, 2, 2, 2) (0, 2, 3, 0);

There are $6$ $G^T$-deformations,
(0, 0, 2, 8) (0, 0, 4, 4) (0, 1, 1, 7) (0, 1, 3, 3) 
(0, 2, 0, 6) (0, 2, 2, 2) 

There are six  $(G,G^T)$-pairs,
[(1/2, 0, 0, 1/2), (0, 1/2, 1/2, 0)] 
[(1/2, 1/2, 0, 0), (0, 0, 1/6, 5/6)] [(1/2, 1/2, 0, 0), (0, 0, 1/3, 2/3)] [(1/2, 1/2, 0, 0), (0, 0, 1/2, 1/2)] [(1/2, 1/2, 0, 0), (0, 0, 2/3, 1/3)] 
[(1/2, 1/2, 0, 0), (0, 0, 5/6, 1/6)]. 

There are 6  $(G^T,G)$-pairs,
[(0, 0, 1/6, 5/6), (1/2, 1/2, 0, 0)] 
[(0, 0, 1/3, 2/3), (1/2, 1/2, 0, 0)] [(0, 0, 1/2, 1/2), (1/2, 1/2, 0, 0)] [(0, 0, 2/3, 1/3), (1/2, 1/2, 0, 0)] [(0, 0, 5/6, 1/6), (1/2, 1/2, 0, 0)] 
[(0, 1/2, 1/2, 0), (1/2, 0, 0, 1/2)] 

$\chi_{Vafa}(W_0/G) = 24$, $h^{1,1}=8 +6+6=20$.

\item 
$|G| = 24$ with generators [(0, 1/4, 0, 3/4), (1/2, 1/4, 1/6, 1/12)], then
$G^T$ has generators [(0, 1/4, 0, 3/4), (1/2, 1/4, 1/6, 1/12)].

There are $7$ $G$-deformations,
(0, 0, 2, 8) (0, 0, 4, 4) (0, 1, 0, 9) (0, 1, 2, 5) 
(0, 1, 4, 1) (0, 2, 0, 6) (0, 2, 2, 2);

there are $7$ $G^T$-deformations,
(0, 0, 2, 8) (0, 0, 4, 4) (0, 1, 0, 9) (0, 1, 2, 5) 
(0, 1, 4, 1) (0, 2, 0, 6) (0, 2, 2, 2). 

There are six  $(G,G^T)$-pairs,
[(0, 1/4, 0, 3/4), (1/2, 0, 1/2, 0)] 
[(0, 1/2, 0, 1/2), (1/2, 0, 1/2, 0)] [(0, 3/4, 0, 1/4), (1/2, 0, 1/2, 0)] [(1/2, 0, 1/2, 0), (0, 1/4, 0, 3/4)] [(1/2, 0, 1/2, 0), (0, 1/2, 0, 1/2)] 
[(1/2, 0, 1/2, 0), (0, 3/4, 0, 1/4)];

There are six  $(G^T,G)$-pairs,
[(0, 1/4, 0, 3/4), (1/2, 0, 1/2, 0)] 
[(0, 1/2, 0, 1/2), (1/2, 0, 1/2, 0)] [(0, 3/4, 0, 1/4), (1/2, 0, 1/2, 0)] [(1/2, 0, 1/2, 0), (0, 1/4, 0, 3/4)] [(1/2, 0, 1/2, 0), (0, 1/2, 0, 1/2)] 
[(1/2, 0, 1/2, 0), (0, 3/4, 0, 1/4)] 

$\chi_{Vafa}(W_0/G) = 24$, $h^{1,1}=7 +7+6=20$.

\end{enumerate}

\subsection{(2, 4, 5, 20)}

$A = (2, 4, 5, 20)$, $\mathbb P  (10, 5, 4, 1)$.

\begin{enumerate}
    \item $G=G^{\max}$, 
$|G| = 40$, generators [(0, 1/4, 0, 3/4), (1/2, 1/4, 1/5, 1/20), 

$G^T=G_0=Z_{20}$, the generator  (1/2, 1/4, 1/5, 1/20).

There are $6$ $G$-deformations,
(0, 0, 1, 16) (0, 0, 2, 12) (0, 0, 3, 8) (0, 2, 0, 10) 
(0, 2, 1, 6) (0, 2, 2, 2). 

There are $10$ $G^T$-deformations,
(0, 0, 1, 16) (0, 0, 2, 12) (0, 0, 3, 8) (0, 1, 0, 15) 
(0, 1, 1, 11) (0, 1, 2, 7) (0, 1, 3, 3) (0, 2, 0, 10) 
(0, 2, 1, 6) (0, 2, 2, 2).

There are four  $(G,G^T)$-pairs, [(1/2, 1/2, 0, 0), (0, 0, 1/5, 4/5)] 
[(1/2, 1/2, 0, 0), (0, 0, 2/5, 3/5)] [(1/2, 1/2, 0, 0), (0, 0, 3/5, 2/5)] [(1/2, 1/2, 0, 0), (0, 0, 4/5, 1/5)]. 

There are $4$  $(G^T,G)$-pairs,
[(0, 0, 1/5, 4/5), (1/2, 1/2, 0, 0)] 
[(0, 0, 2/5, 3/5), (1/2, 1/2, 0, 0)] [(0, 0, 3/5, 2/5), (1/2, 1/2, 0, 0)] [(0, 0, 4/5, 1/5), (1/2, 1/2, 0, 0)] 

$\chi_{Vafa}(W_0/G) = 24$, $h^{1,1}=6 +10+4=20$.

\end{enumerate}

\subsection{(2, 3, 12, 12)}

$A = (2, 3, 12, 12)$, $\mathbb P (6, 4, 1, 1)$.

\begin{enumerate}
    \item $G=G^{\max}$, $
|G| = 72$, generators  [(0, 0, 1/12, 11/12), (1/2, 1/3, 1/12, 1/12)]; $G^T=G_0=Z_{12}$, the generator  (1/2, 1/3, 1/12, 1/12).

There are $2$ $G$-deformations,
(0, 0, 6, 6) (0, 1, 4, 4). 

There are $18$ $G^T$-deformations,
(0, 0, 2, 10) (0, 0, 3, 9) (0, 0, 4, 8) (0, 0, 5, 7) 
(0, 0, 6, 6) (0, 0, 7, 5) (0, 0, 8, 4) (0, 0, 9, 3) 
(0, 0, 10, 2) (0, 1, 0, 8) (0, 1, 1, 7) (0, 1, 2, 6) 
(0, 1, 3, 5) (0, 1, 4, 4) (0, 1, 5, 3) (0, 1, 6, 2) 
(0, 1, 7, 1) (0, 1, 8, 0). 

There are no Roan's pairs. 

$\chi_{Vafa}(W_0/G) = 24$, $h^{1,1}=2 +18+0=20$.

\item 
$|G| = 24$, the generators  [(0, 0, 1/4, 3/4), (1/2, 1/3, 1/12, 1/12)]; 
$G^T$ has generators  [(0, 0, 1/6, 5/6), (1/2, 1/3, 1/12, 1/12)]. 

There are $10$ $G$-deformations,
(0, 0, 2, 10) (0, 0, 4, 8) (0, 0, 6, 6) (0, 0, 8, 4) 
(0, 0, 10, 2) (0, 1, 0, 8) (0, 1, 2, 6) (0, 1, 4, 4) 
(0, 1, 6, 2) (0, 1, 8, 0). 

There are $6$ $G^T$-deformations, 
(0, 0, 3, 9) (0, 0, 6, 6) (0, 0, 9, 3) (0, 1, 1, 7) 
(0, 1, 4, 4) (0, 1, 7, 1).

There are four  $(G,G^T)$-pairs,
[(1/2, 0, 0, 1/2), (0, 1/3, 2/3, 0)] 
[(1/2, 0, 0, 1/2), (0, 2/3, 1/3, 0)] [(1/2, 0, 1/2, 0), (0, 1/3, 0, 2/3)] [(1/2, 0, 1/2, 0), (0, 2/3, 0, 1/3)].

There are four  $(G^T,G)$-pairs,
[(0, 1/3, 0, 2/3), (1/2, 0, 1/2, 0)] 
[(0, 1/3, 2/3, 0), (1/2, 0, 0, 1/2)] [(0, 2/3, 0, 1/3), (1/2, 0, 1/2, 0)] [(0, 2/3, 1/3, 0), (1/2, 0, 0, 1/2)].

$\chi_{Vafa}(W_0/G) = 24$, $h^{1,1}=10 +6+4=20$.

\end{enumerate}

\subsection{(2, 3, 10, 15)}
$A = (2, 3, 10, 15)$, $\mathbb P (15, 10, 3, 2)$.

$G=G^{\max}$, $
|G| = 30$, the generator  (1/2, 1/3, 1/10, 1/15), 
$G^T=G$. 

There are $8$ $G$-deformations,
(0, 0, 2, 12) (0, 0, 4, 9) (0, 0, 6, 6) (0, 0, 8, 3) 
(0, 1, 0, 10) (0, 1, 2, 7) (0, 1, 4, 4) (0, 1, 6, 1);

there are $8$ $G^T$-deformations,
(0, 0, 2, 12) (0, 0, 4, 9) (0, 0, 6, 6) (0, 0, 8, 3) 
(0, 1, 0, 10) (0, 1, 2, 7) (0, 1, 4, 4) (0, 1, 6, 1).

There are four  $(G,G^T)$-pairs,
[(0, 1/3, 0, 2/3), (1/2, 0, 1/2, 0)] 
[(0, 2/3, 0, 1/3), (1/2, 0, 1/2, 0)] [(1/2, 0, 1/2, 0), (0, 1/3, 0, 2/3)] [(1/2, 0, 1/2, 0), (0, 2/3, 0, 1/3)];

there are four  $(G^T,G)$-pairs,
[(0, 1/3, 0, 2/3), (1/2, 0, 1/2, 0)] 
[(0, 2/3, 0, 1/3), (1/2, 0, 1/2, 0)] [(1/2, 0, 1/2, 0), (0, 1/3, 0, 2/3)] [(1/2, 0, 1/2, 0), (0, 2/3, 0, 1/3)]. 
    
$\chi_{Vafa}(W_0/G) = 24$, $h^{1,1}=8 +8+4=20$.

\subsection{(2, 3, 9, 18)}

$A = (2, 3, 9, 18)$, $\mathbb P (9, 6, 2, 1)$.

$G=G^{\max}$, $
|G| = 54$, the generators  [(0, 0, 1/9, 8/9), (1/2, 1/3, 1/9, 1/18)], $G^T=G_0=Z_{18}$, the generator  (1/2, 1/3, 1/9, 1/18).

There are $4$ $G$-deformations,
(0, 0, 3, 12) (0, 0, 6, 6) (0, 1, 1, 10) (0, 1, 4, 4);

There are $14$ $G^T$-deformations,
(0, 0, 1, 16) (0, 0, 2, 14) (0, 0, 3, 12) (0, 0, 4, 10) 
(0, 0, 5, 8) (0, 0, 6, 6) (0, 0, 7, 4) (0, 1, 0, 12) 
(0, 1, 1, 10) (0, 1, 2, 8) (0, 1, 3, 6) (0, 1, 4, 4) 
(0, 1, 5, 2) (0, 1, 6, 0).

There are two $(G,G^T)$-pairs,
[(0, 1/3, 2/3, 0), (1/2, 0, 0, 1/2)] 
[(0, 2/3, 1/3, 0), (1/2, 0, 0, 1/2)];

and there are two  $(G^T,G)$-pairs,
[(1/2, 0, 0, 1/2), (0, 1/3, 2/3, 0)] 
[(1/2, 0, 0, 1/2), (0, 2/3, 1/3, 0)].

$\chi_{Vafa}(W_0/G) = 24$, $h^{1,1}=4 +14+2=20$.

\subsection{(2, 3, 8, 24)}

$A = (2, 3, 8, 24)$, $\mathbb P (12, 8, 3, 1)$. 

$G=G^{\max}$, $
|G| = 48$, the generators [(0, 0, 1/8, 7/8), (1/2, 1/3, 1/8, 1/24)], $G^T=G_0=Z_{24}$, the generator  (1/2, 1/3, 1/8, 1/24)]. 

There are $6$ $G$-deformations,
(0, 0, 2, 18) (0, 0, 4, 12) (0, 0, 6, 6) (0, 1, 0, 16) 
(0, 1, 2, 10) (0, 1, 4, 4).

There are $12$ $G^T$-deformations,
(0, 0, 1, 21) (0, 0, 2, 18) (0, 0, 3, 15) (0, 0, 4, 12) 
(0, 0, 5, 9) (0, 0, 6, 6) (0, 1, 0, 16) (0, 1, 1, 13) 
(0, 1, 2, 10) (0, 1, 3, 7) (0, 1, 4, 4) (0, 1, 5, 1).

There are two $(G,G^T)$-pairs,
[(1/2, 0, 1/2, 0), (0, 1/3, 0, 2/3)] 
[(1/2, 0, 1/2, 0), (0, 2/3, 0, 1/3)], and there are two $(G^T,G)$-pairs, 
[(0, 1/3, 0, 2/3), (1/2, 0, 1/2, 0)] 
[(0, 2/3, 0, 1/3), (1/2, 0, 1/2, 0)]. 

$\chi_{Vafa}(W_0/G) = 24$, $h^{1,1}=6 +12+2=20$.

\subsection{(2, 3, 7, 42)}
$A = (2, 3, 7, 42)$, $\mathbb P (21, 14, 6, 1)$. 
$G=G^{\max}$, $
|G| = 42$, the generator (1/2, 1/3, 1/7, 1/42), $G^T=G_0$, the generator  (1/2, 1/3, 1/7, 1/42).

There are $10$ $G$-deformations,
(0, 0, 1, 36) (0, 0, 2, 30) (0, 0, 3, 24) (0, 0, 4, 18) 
(0, 0, 5, 12) (0, 1, 0, 28) (0, 1, 1, 22) (0, 1, 2, 16) 
(0, 1, 3, 10) (0, 1, 4, 4).

There are the same $10$ $G^T$-deformations,
(0, 0, 1, 36) (0, 0, 2, 30) (0, 0, 3, 24) (0, 0, 4, 18) 
(0, 0, 5, 12) (0, 1, 0, 28) (0, 1, 1, 22) (0, 1, 2, 16) 
(0, 1, 3, 10) (0, 1, 4, 4).

There are no Roan's pairs. 

$\chi_{Vafa}(W_0/G) = 24$, $h^{1,1}=10 +10+0=20$.







\end{document}